\documentclass[11pt]{article}

\newcommand{\later}[1]{{}}
\newcommand{\old}[1]{{}}
\long\def\ignore#1{}
\def\cut{\ignore}

\cut{28 figure files:
s-9-C.pdf       s-11-B.pdf    s-8-1.pdf    pentagon.pdf greedy.pdf
example.pdf     P0.pdf        s-10-A.pdf   s-11-D.pdf   s-8-3.pdf
P1.pdf          s-10-B.pdf    s-9-A1.pdf   P2.pdf       s-10-C.pdf
s-11-E1.pdf     s-9-A61.pdf   s23ub.pdf    P3.pdf       s-10-D1.pdf
s-9-B1.pdf      P4.pdf        s-11-A1.pdf  s23.pdf      convex.pdf
nonconvex.pdf s25.pdf      s26.pdf
} 

\usepackage{amsmath}
\usepackage{amsthm}
\usepackage{amssymb}
\usepackage{amsfonts}
\usepackage{float}
\usepackage{graphicx}
\usepackage[small]{caption}
\usepackage{listings}
\usepackage{color}
\usepackage{tcolorbox}
\usepackage{enumitem}
\usepackage{marvosym}
\usepackage{url}
\usepackage{pifont}

\definecolor{dkgreen}{rgb}{0,0.4,0}
\definecolor{gray}{rgb}{0.5,0.5,0.5}
\definecolor{mauve}{rgb}{0.58,0,0.82}
\definecolor{orange}{rgb}{0.8,0.1,0}

\newtheorem{theorem}{Theorem}
\newtheorem{lemma}{Lemma}

\newtheorem{question}{Problem}

\newcommand{\eps}{\varepsilon}

\newcommand{\etal}{{et~al.}}
\newcommand{\ie}{{i.e.}}
\newcommand{\eg}{{e.g.}}

\def\deg{\texttt{deg}}
\def\up{\texttt{upper}}
\def\lo{\texttt{lower}}

\begin{document}

\title{Lower bounds on the dilation of plane spanners\footnote{%
A preliminary version in: Proceedings of the 2nd International Conference on Algorithms
and Discrete Applied Mathematics, Thiruvananthapuram, India, 
Feb. 2016, vol. $9602$ of LNCS.}}

\author{%
Adrian Dumitrescu\\
\small Department of Computer Science\\[-0.8ex]
\small University of Wisconsin--Milwaukee\\[-0.8ex]
\small Milwaukee, WI, USA\\
\small\tt \Email{ dumitres@uwm.edu}\\
\and
Anirban Ghosh\\
\small Department of Computer Science\\[-0.8ex]
\small University of Wisconsin--Milwaukee\\[-0.8ex]
\small Milwaukee, WI, USA\\
\small\tt \Email{ anirban@uwm.edu}\\
}

\date{\today}
\maketitle

\begin{abstract}
  (I) We exhibit a set of 23 points in the plane that has dilation
  at least $1.4308$, improving the previous best lower bound of $1.4161$
  for the worst-case dilation of plane spanners. 

  \smallskip
  (II) For every $n\geq13$, there exists an $n$-element point set $S$ such
  that the degree~$3$ dilation of $S$ equals $1+\sqrt{3}=2.7321\ldots$
  in the domain of plane geometric
  spanners. In the same domain, we show that for every $n\geq6$,
  there exists a an $n$-element point set $S$ such that the degree~$4$
  dilation of $S$ equals $1 + \sqrt{(5-\sqrt{5})/2}=2.1755\ldots$
  The previous best lower bound of $1.4161$ holds for any degree. 

\smallskip
  (III) For every $n\geq6 $, there exists an $n$-element point set $S$ such
that the stretch factor of the greedy triangulation of $S$ is at least $2.0268$. 

\end{abstract}

\section{Introduction} \label{sec:intro}

Given a set of points $P$ in the Euclidean plane,
a \emph{geometric graph} on $P$ is a weighted graph $G = (V,E)$ where
$V = P$ and an edge $uv \in E$ is the line segment with endpoints
$u,v \in V$ weighted by the Euclidean distance $|uv|$ between
them. For $t\geq 1$, a geometric graph $G$ is a \emph{t-spanner}, if
for every pair of vertices $u,v$ in $V$, the length of the shortest
path $\pi_G(u,v)$ between them in $G$ is at most $t$ times $|uv|$,
\ie, $\forall u,v \in V,\,|\pi_G(u,v)| \leq t |uv|$. A complete
geometric graph on a set of points is a 1-spanner. Where there is no
necessity to specify $t$, we use the term \emph{geometric spanner}.
A geometric spanner $G$ is \emph{plane} if no two edges
in $G$ cross. In this paper we only consider plane geometric spanners.
A geometric spanner of degree at most $k$ is referred to as a
\emph{degree $k$ geometric spanner}. 

Given a geometric spanner $G=(V,E)$, the \emph{vertex dilation}
or \emph{stretch factor} of $u,v\in V$, denoted $\delta_G(u,v)$,
is defined as $\delta_G(u,v) =  |\pi_G(u,v)| / |uv|$.
When $G$ is clear from the context, we simply write $\delta(u,v)$. 
The \emph{vertex dilation} or \emph{stretch factor} of $G$,
denoted $\delta(G)$,  is defined as  
$\delta(G) = \sup_{u,v \in V} \delta_G(u,v). $ The terms
\emph{graph theoretic dilation} and \emph{spanning ratio} are also used
in the literature. Refer to~\cite{EGK06,KG92,NS07} for such definitions. 

Given a point set $P$, let the \emph{dilation} of $P$, denoted by $\delta_0(P)$,
be the minimum stretch factor of a plane geometric graph (equivalently, triangulation)
on vertex set $P$; see~\cite{Mu04}. 
Similarly, let the \emph{degree $k$ dilation} of $P$, denoted by $\delta_0(P,k)$,
be the minimum stretch factor of a plane geometric graph of degree at most $k$
on vertex set $P$. Clearly, $\delta_0(P,k) \geq \delta_0(P)$
holds for any $k$. Furthermore, $\delta_0(P,j) \geq \delta_0(P,k)$ holds for any $j<k$.
(Note that the term \emph{dilation} has been also used with different meanings in
the literature, see for instance~\cite{PBMS13,KKP15}.)

In the last few decades, great progress has been made in the field
of geometric spanners; for an overview refer to~\cite{GK07,NS07}.
Common goals include constructions of low stretch factor geometric spanners
that have few edges, bounded degree and so on. A survey of open
problems in this area along with existing results can be found in~\cite{PBMS13}.
Geometric spanners find their applications in the areas of robotics,
computer networks, distributed systems and many others.
Refer to~\cite{AKK+08,ADD+93,ABC+08,CDNS95,EKLL04,LL92} for various algorithmic
results.

The existence of plane $t$-spanners for some constant $t>1$
(with no restriction on degree)
was first investigated by Chew~\cite{Ch89} in the 80s.
He showed that it is always possible to construct a plane 2-spanner with $O(n)$ edges
on a set of $n$ points; he also observed that every plane geometric graph embedded
on the $4$ points placed at the vertices of a square has stretch factor at least $\sqrt{2}$.
This was the best lower bound on the worst-case dilation of plane spanners for
almost $20$ years until it was shown by Mulzer~\cite{Mu04} using a computer program
that every triangulation of a regular $21$-gon has stretch factor at least
$( 2\sin \frac{\pi}{21} + \sin \frac{5\pi}{21} +
\sin \frac{3\pi}{21}) / \sin \frac{10\pi}{21} = 1.4161\ldots$
Henceforth, it was posed as an open problem by
Bose and Smid~\cite[Open Problem 1]{PBMS13}
(as well as by Kanj in his survey~\cite[Open Problem 5]{Ka13}): 
``\emph{What is the best lower bound on the spanning ratio of plane geometric graphs?
Specifically, is there a $t > \sqrt{2.005367532} \approx 1.41611\ldots$
and a point set $P$, such that every triangulation of $P$ has spanning ratio at least $t$?}''.
We give a positive answer to the second question by showing that a set $S$ of 23
points placed at the vertices of a regular 23-gon, has dilation $\delta_0(S) \geq
(2\sin\frac{2\pi}{23}+\sin\frac{8\pi}{23})/\sin\frac{11\pi}{23}=1.4308\ldots$

The problem can be traced back to a survey written by
Eppstein~\cite[Open Problem 9]{Epp00}:
``\emph{What is the worst case dilation of the minimum dilation triangulation?}".
The point set $S$ also provides a partial answer for this question.
From the other direction, the current best upper bound of $1.998$
was proved by Xia~\cite{Xia13} using Delaunay triangulations.
Note that this bound is only slightly better than the bound of $2$
obtained by Chew~\cite{Ch89} in the 1980s.
For previous results on the upper bound refer to~\cite{CKX11,DJ89,DFS90,KG92}.

The design of low degree plane spanners is of great interest to geometers.
Bose~\etal~\cite{BGS05} were the first to show that there always exists
a plane $t$-spanner of degree at most 27 on any set of points in the Euclidean plane
where $t \approx 10.02$. The result was subsequently improved
in~\cite{BGHP10,BKPX15,BCC12,BSX09,KP08,LW04} in terms of degree.
Recently, Kanj~\etal~\cite{KPT16} showed that $t=20$ can be achieved with degree $4$.
However, the question whether the degree can be reduced to $3$ remains
open at the time of this writing. If one does not insist on having a plane spanner,
Das~\etal~\cite{DH96} showed that degree~$3$ is achievable.
While numerous papers have focused on upper bounds on the dilation of bounded degree
plane spanners, not much is known about lower bounds. In this paper,
we explore this direction and provide new lower bounds for unrestricted degrees
and when degrees~$3$ and~$4$ are imposed.

A  \emph{greedy triangulation} of a finite point set  $P$ is constructed in
the following way: starting with an empty set of edges $E$, repeatedly
add edges to $E$ in non-decreasing order of length as long as edges in $E$ are noncrossing.
Bose~\etal~\cite{BLS07} have showed that the greedy triangulation is a $t$-spanner, where 
$t = 8( \pi-\alpha )^2/(\alpha^2 \sin^2( \alpha/ 4 )) \approx 11739.1$ and $\alpha = \pi/6$.  
Here we obtain a worst-case lower bound of $2.0268$; in light of computational experiments
we carried out, we believe that the aforementioned upper bound is very far from the truth.

\paragraph{Related work.} If $S_n$ is the set of $n$ vertices of a regular $n$-gon,
Mulzer~\cite{Mu04} showed that \linebreak
$ 1.3836 \ldots = \sqrt{2-\sqrt{3}} + \sqrt{3}/2 \leq \delta_0(S_n)
\leq 0.471\pi /\sin 0.471\pi = 1.4858\ldots$, for every $n\geq74$;
the upper bound holds for every $n\geq 3$.  
Amarnadh and Mitra~\cite{AM06} have shown that in the case of a
cyclic polygon (a polygon whose vertices are co-circular), the stretch factor of
any \emph{fan} triangulation (\ie, with a vertex of degree $n-1$) is $\leq 1.4846$.

As mentioned earlier, low degree plane spanners for general point sets
have been studied in~\cite{BGHP10,BCC12,BGS05,BSX09,KP08,LW04}. 
The construction of low degree plane spanners for the infinite
square and hexagonal lattices has been recently investigated in~\cite{DG16}. 

Bose~\etal~\cite{BDL+11} presented a finite convex point set
for which there is a Delaunay triangulation whose stretch factor is at
least $1.581 > \pi/2$, thereby disproving a widely believed $\pi/2$
upper bound conjectured by Chew~\cite{Ch89}.
They also showed that this lower bound can be slightly
raised to $1.5846$ if the point set need not be convex.
This lower bound for non-convex point sets has been further improved to
$1.5932$ by Xia and Zhang~\cite{XZ11}.

Klein~\etal~\cite{KKP15} proved the following interesting structural result.
Let $S$ be a finite set of points in the plane. 
Then either $S$ is a subset of one of the well-known sets of points whose triangulation is
unique and has dilation $1$, or there exists a number $\Delta(S)> 1$ such that each finite
plane graph containing $S$ among its vertices has dilation at least $\Delta(S)$. 

Cheong~\etal~\cite{CHL08} showed that for every $n\geq 5$, there are sets of $n$
points in the plane that do not have a minimum-dilation spanning tree
without edge crossings and that $5$ is minimal with this property.
They also showed that given a set $S$ of $n$ points with
integer coordinates in the plane and a rational dilation $t>1$, it is
NP-hard to decide whether a spanning tree of $S$ with dilation at most
$t$ exists, regardless if edge crossings are allowed or not. 

Knauer and Mulzer~\cite{KM05} showed that for each edge $e$ of a
\emph{minimum dilation triangulation} of a point set,
at least one of the two half-disks of diameter about $0.2|e|$
on each side of $e$ and centered at the midpoint of $e$
must be empty of points\footnote{Their result inaccurately
states that the entire disk of that diameter is an exclusion region.}.

When the stretch factor (or dilation) is measured over all pairs of
points on edges or vertices of a plane graph $G$ (rather than only over pairs of vertices) 
one arrives at the concept of  \emph{geometric dilation} of $G$; see~\cite{DEG+06,EGK06}.

\paragraph{Our results.}
(I) Let $S$ be a set of 23 points placed at the vertices of a regular 23-gon. Then,
  $\delta_0(S) =(2\sin\frac{2\pi}{23}+\sin\frac{8\pi}{23})/\sin\frac{11\pi}{23}=1.4308\ldots$
(Theorem~\ref{thm-s23}, Section~\ref{sec:lowerbound}).
This improves the previous bound of 
     $(2\sin \frac{\pi}{21} + \sin \frac{5\pi}{21} +\sin \frac{3\pi}{21} )
     / \sin \frac{10\pi}{21} = 1.4161\ldots$,
     due to Mulzer~\cite{Mu04}, on the worst case dilation of plane spanners.     

\smallskip
(II) (a) For every $n\geq 13$, there exists a set
$S$ of $n$ points such that $\delta_0(S,3) \geq 1 +\sqrt{3}
= 2.7321\ldots$ (Theorem~\ref{thm:deg3}, Section~\ref{sec:deg3and4}).
(b) For every $n\geq 6$, there exists a set $S$ of $n$ points such that
 $\delta_0(S,4) \geq 1 + \sqrt{(5-\sqrt{5})/2} = 2.1755\ldots$ 
 (Theorem~\ref{thm:deg4}, Section~\ref{sec:deg3and4}). 
The previous best lower bound of 
$(2\sin \frac{\pi}{21} + \sin \frac{5\pi}{21} +\sin \frac{3\pi}{21} )
/ \sin \frac{10\pi}{21}=1.4161\ldots$, due to Mulzer~\cite{Mu04}
 holds for any degree. Here we sharpen it for degrees~$3$ and~$4$. 

 \smallskip
 (III)  For every $n\geq6 $, there exists a set $S$ of $n$ points
  such that the stretch factor of the greedy triangulation of $S$ is at least $2.0268$.

\paragraph{Notations and assumptions.}
Let $P$ be a planar point set and $G=(V,E)$ be a plane geometric graph
on vertex set $P$.
For $p,q \in P$, $pq$ denotes the connecting segment and
$|pq|$ denotes its Euclidean length.
The degree of a vertex (point) $p \in P$ is denoted by $\deg(p)$.
For a specific point set $P=\{p_1,\ldots,p_n\}$,
we denote a path in $G$ consisting of vertices in the order
$p_{i},p_{j},p_{k},\ldots$ using $\rho(i,j,k,\ldots)$ and by
$|\rho(i,j,k,\ldots)|$ its total Euclidean length.
The graphs we construct have the property that no edge
contains a point of $P$ in its interior.

\section{A new lower bound on the dilation of plane spanners}
\label{sec:lowerbound}

In this section, we show that the set $S=\{s_0,\ldots,s_{22}\}$ of $23$
points  placed at the vertices of a regular $23$-gon has
dilation $\delta_0(S) \geq
(2\sin\frac{2\pi}{23}+\sin\frac{8\pi}{23})/\sin\frac{11\pi}{23}=1.4308\ldots$
(see Fig.~\ref{fig:s23}). Assume that the points lie on a circle of unit radius. 
We first present a theoretical proof showing that
$\delta_0(S) \geq  (\sin\frac{2\pi}{23}+\sin\frac{4\pi}{23}+\sin\frac{5\pi}{23})
/\sin\frac{11\pi}{23}=1.4237\ldots$;
we then raise the bound to
$\delta_0(S) \geq (2\sin\frac{2\pi}{23}+\sin\frac{8\pi}{23})/\sin\frac{11\pi}{23}=1.4308\ldots$
using a computer program.
The result obtained by the program is tight as there exists a triangulation of $S$
(see Fig.~\ref{fig:s23}\,(right)) with stretch factor exactly
$(2\sin\frac{2\pi}{23}+\sin\frac{8\pi}{23})/\sin\frac{11\pi}{23}=1.4308\ldots$

Define the \emph{convex hull length} of a chord $s_i s_j \in S$
as  $\mu(i,j) = \min(|i-j|,23-|i-j|).$ Observe that $1 \leq \mu(i,j) \leq 11$.
Since triangulations are maximal planar graphs, 
we only consider triangulations of $S$ while computing $\delta_0(S)$;
in particular, every edge of the convex hull of $S$ is present.
Note that there are $C_{21}=24,466,267,020$ triangulations of $S$.
Here $C_n = \frac{1}{n+1}\binom{2n}{n}$ is the $n^{th}$ \emph{Catalan number}
and there are $C_n$ ways to triangulate a convex polygon with $n+2$ vertices. 

If $s_i,s_j \in S$, then $|s_is_j| = 2\sin \frac{\mu(i,j)\pi}{23}$.
Consider a shortest path connecting $s_i,s_j \in S$ consisting of $k$ edges 
with convex hull lengths $n_1,\ldots,n_k$; its length is
$|\rho(i,\ldots,j)|= 2 \sum_{h=1}^{k} \sin \frac{n_h \pi}{23}$.
Let $\lambda=\mu(i,j)$ and
\begin{equation} \label{eq:g}
g(\lambda,n_1,\ldots,n_k) = \frac{|\rho(i,\ldots,j)|}{ |s_is_j|} =
\frac{ \sum_{h=1}^{k} \sin \frac{n_h \pi}{23}}{\sin\frac{\lambda \pi}{23}}.
\end{equation}
We will use $\lambda=11$ in all subsequent proofs of this section and therefore we set 
\begin{equation} \label{eq:f}
  f(n_1,\ldots,n_k) :=g(11,n_1,\ldots,n_k).
\end{equation}

Various values of $f$, as given by~\eqref{eq:g} and~\eqref{eq:f},
will be repeatedly used in lower-bounding the stretch factor of point pairs
in specific configurations, \ie, when some edges are assumed to be present.
Observe that $f$ is a symmetric function that can be easily computed (tabulated)
at each tuple $n_1,\ldots,n_k$; see Table~\ref{table:f}. 
\begin{figure}[hbtp]
\centering
\includegraphics[scale=0.5]{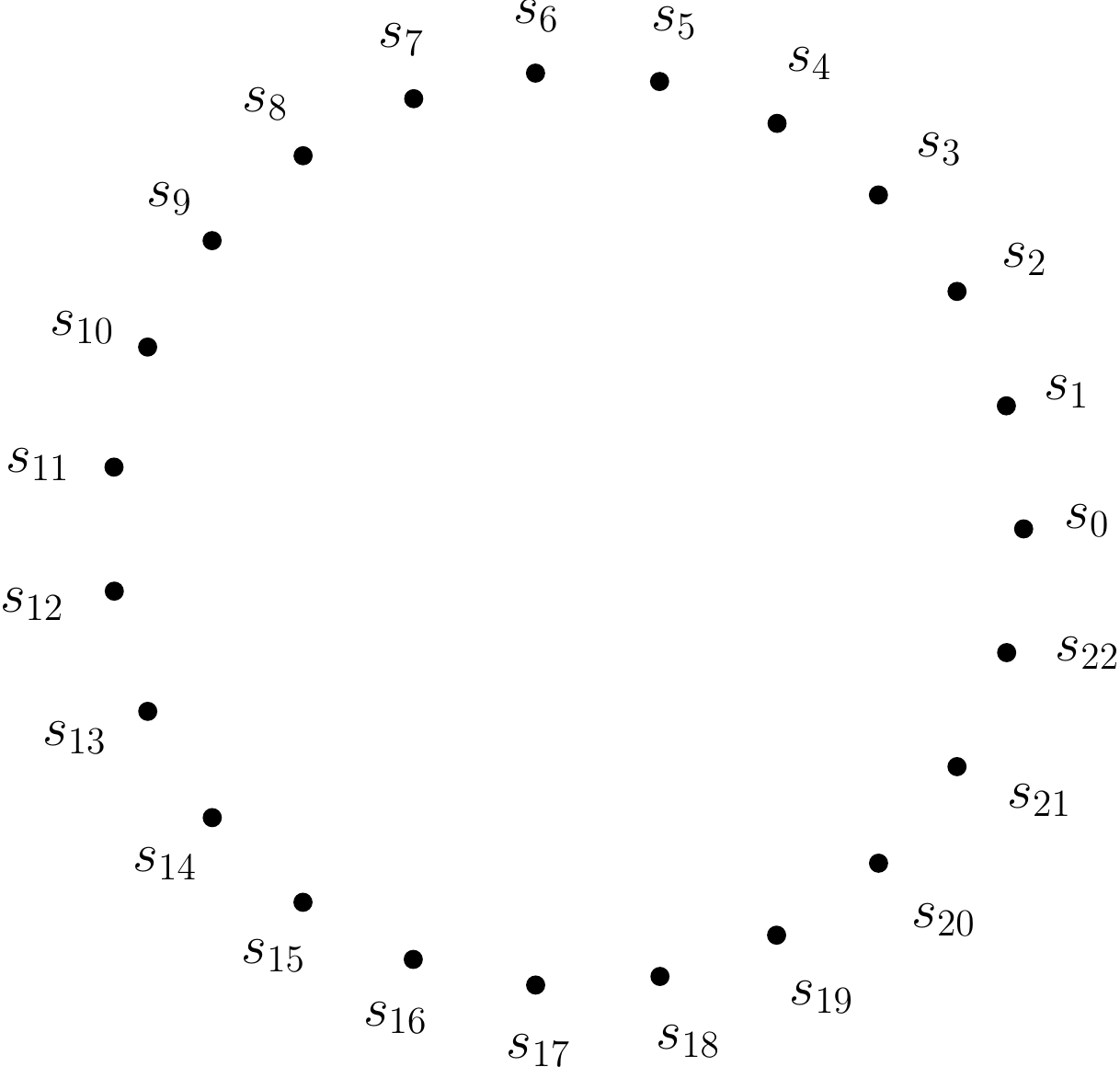}
\hspace{15mm}
\includegraphics[scale=0.5]{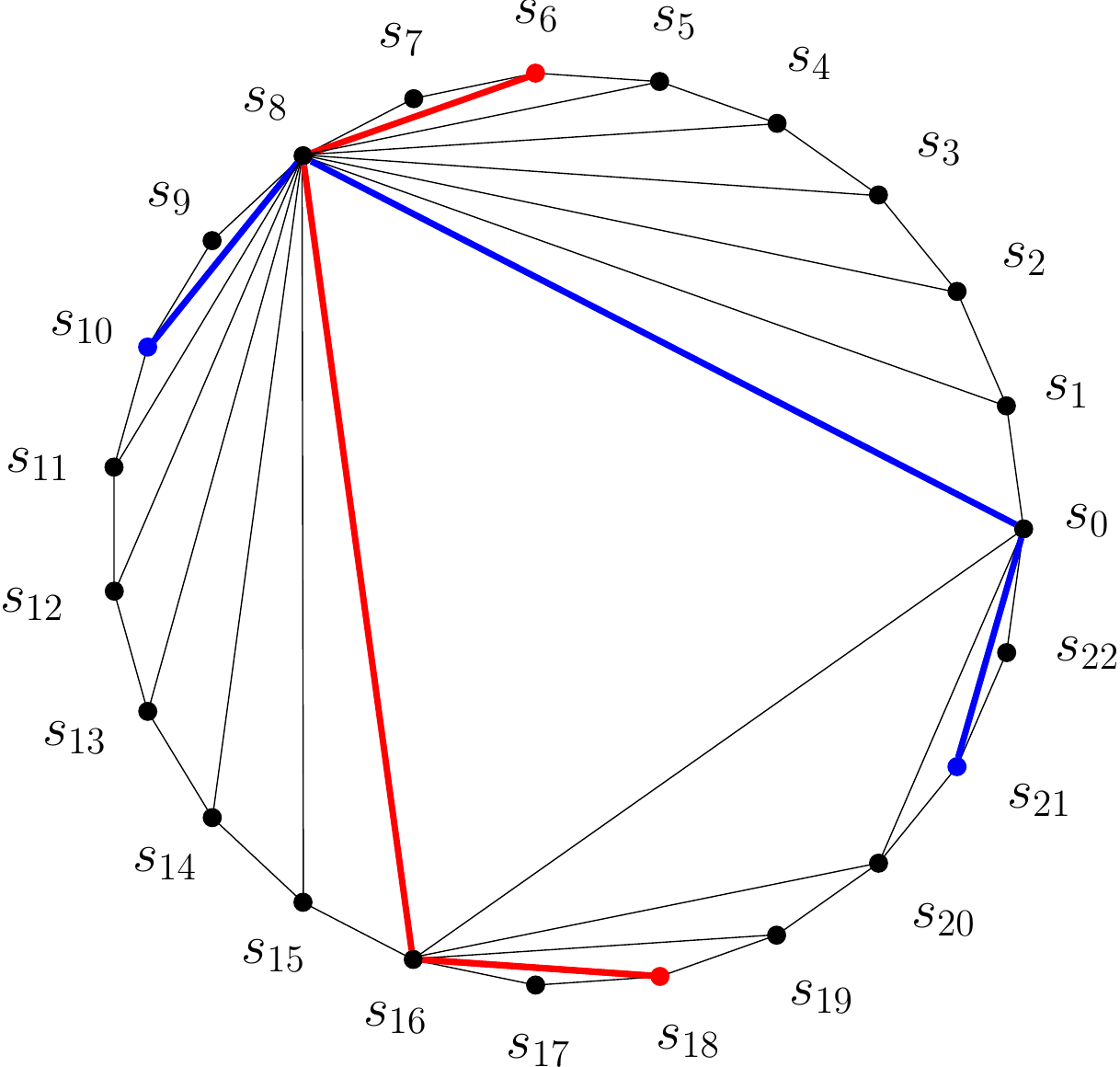}
\caption{Left: The set $S$ of 23 points placed at the
  vertices of a regular 23-gon. Right: A triangulation of $S$ with stretch factor
  $(2\sin\frac{2\pi}{23}+\sin\frac{8\pi}{23})/\sin\frac{11\pi}{23}=1.4308\ldots$,
  which is achieved by the detours for the pairs $s_{10},s_{21}$ and $s_{6},s_{18}$.
  The shortest paths connecting the pairs are shown in blue and red, respectively.} 
\label{fig:s23}
\end{figure}

\begin{table}[ht]
\centering
\begin{tabular}{|l l | c || c l | c|}
	\hline
	&$f(4,7)$  & $1.3396\ldots$ &\ding{118}&$f(2,2,8)$ &  $1.4308\ldots$\\
	&$f(5,6)$  & $1.3651\ldots$ &\ding{118}&$f(3,3,5)$ &  $1.4312\ldots$\\
	\ding{118}&$f(5,7)$ & $1.4514\ldots$ &\ding{118}&$f(3,4,4)$ &  $1.4409\ldots$\\
	\ding{118}&$f(6,6)$ & $1.4650\ldots$ &\ding{118}&$f(1,4,7)$ &  $1.4761\ldots$\\
	\cline{1-3}
	&$f(2,3,6)$ &  $1.4023\ldots$ &\ding{118}&$f(2,3,7)$ &  $1.4886\ldots$ \\
	&$f(1,5,5)$ &  $1.4061\ldots$ &\ding{118}&$f(3,3,6)$ &  $1.5312\ldots$ \\
	\cline{4-6}
	\ding{118}&$f(2,4,5)$ &  $1.4237\ldots$ &\ding{118}&$f(1,1,4,5)$ &  $1.4263\ldots$ \\
	\ding{118}&$f(1,3,8)$ &  $1.4257\ldots$ &\ding{118}&$f(1,2,3,5)$ &  $1.4388\ldots$\\
    \hline
\end{tabular}
\caption{Relevant values of $f(n_1,\ldots,n_k)$ as required by the proofs in this section. 
Values used explicitly in the proofs are marked using \ding{118}.}
\label{table:f}
\end{table}

Given a chord $s_0 s_i$, let $\lo(s_0s_i) = \{s_{i+1},\ldots,s_{22}\}$ and
$\up(s_0s_i)  = \{s_1,\ldots s_{i-1}\}$. The range of possible convex hull
lengths of the longest chord in a triangulation of $S$ is given by the following.

\begin{lemma}\label{lem:ell}
If $\ell$ is the convex hull length of the longest chord in a triangulation of $S$, 
then $\ell \in \{8,9,10,11\}$. 
\end{lemma}
\begin{proof}
We clearly have $\ell \geq 2$. 
Since $S$ is symmetric, we can assume that $s_0 s_\ell$ is the longest chord.
 Since $\mu(i,j) \leq 11$ for any $0 \leq i,j \leq 22$, we have $\ell \leq 11$. 
 Suppose for contradiction that $2 \leq \ell \leq 7$. Then $s_0 s_\ell$ is an edge
 of some triangle $\Delta{s_0 s_\ell s_m}$, where $\ell+1 \leq m \leq 22$.
 In particular,
\begin{equation} \label{eq:ell}
 \mu(0,\ell), \mu(\ell,m), \mu(0,m) \leq \ell \leq 7.
\end{equation} 

If $m \leq 11$, then $\mu(0,m) = \min(m,23-m)=m \geq \ell+1$, a contradiction to
$\ell$'s maximality. Assume now that $m \geq 12$; then
$\mu(0,m) = 23-m \leq \ell$, since $\ell$ is the length of a longest chord. 
It follows that $m \geq 23 -\ell \geq 23-7=16$. 
If $m-\ell \leq 11$, then $\mu(\ell,m) = m-\ell  \geq 16-7=9$, a contradiction to~\eqref{eq:ell}.
If $m-\ell \geq 12$, then $\mu(\ell,m) =23 -(m-\ell) =23 -m +\ell \geq \ell+1$,
a contradiction to $\ell$'s maximality. 
Consequently, we have $8 \leq \ell \leq 11$, as required.
\end{proof}

\paragraph{Proof outline.} 
For every $\ell \in \{ 8,9,10,11\}$, if the longest chord in a triangulation $T$ 
has length $\ell$, we show that $\delta(T) \geq f(2,4,5) = 1.4237\ldots$
Assuming that $s_0s_\ell$ is a longest chord,
we consider the triangle with base $s_0s_\ell$ and third vertex
in $\up(s_0s_\ell)$ or $\lo(s_0s_\ell)$, depending on $\ell$.
For each such triangle, we show that if the edges of
the triangle along with the convex hull edges of $S$ are present, then
in any resulting triangulation there is a pair whose
stretch factor is at least $f(2,4,5) = 1.4237\ldots$ 
Essentially, the long chords
act as obstacles which contribute to long detours for some point pairs.
In four subsequent lemmas, we consider the convex hull lengths $8,9,10,11$
(from Lemma~\ref{lem:ell}) successively. 

In some arguments, we consider a \emph{primary pair} $s_i,s_j$, and possible
shortest paths between the two vertices. We show that if certain  
intermediate vertices are present in $\pi(s_i,s_j)$, then $\delta(s_i,s_j) \geq f(2,4,5)$.
Otherwise if certain long edges are present in $\pi(s_i,s_j)$,
then $\delta(s_u,s_v) \geq f(2,4,5)$, where $s_u,s_v$ is a \emph{secondary pair}.
In the figures, wherever required, we use circles and squares to mark the primary
and secondary pairs, respectively (see for instance Fig.~\ref{fig:s23_8}).
In some of the cases, a primary pair suffices in the argument,
\ie, no secondary pair is needed.

\begin{lemma}{\label{lem:s_23.8}}
If $\ell=8$, then $\delta(T) \geq f(2,4,5) = 1.4237\ldots$
\end{lemma}
\begin{proof} 
\begin{figure}[hbtp]
\centering
\includegraphics[scale=0.5]{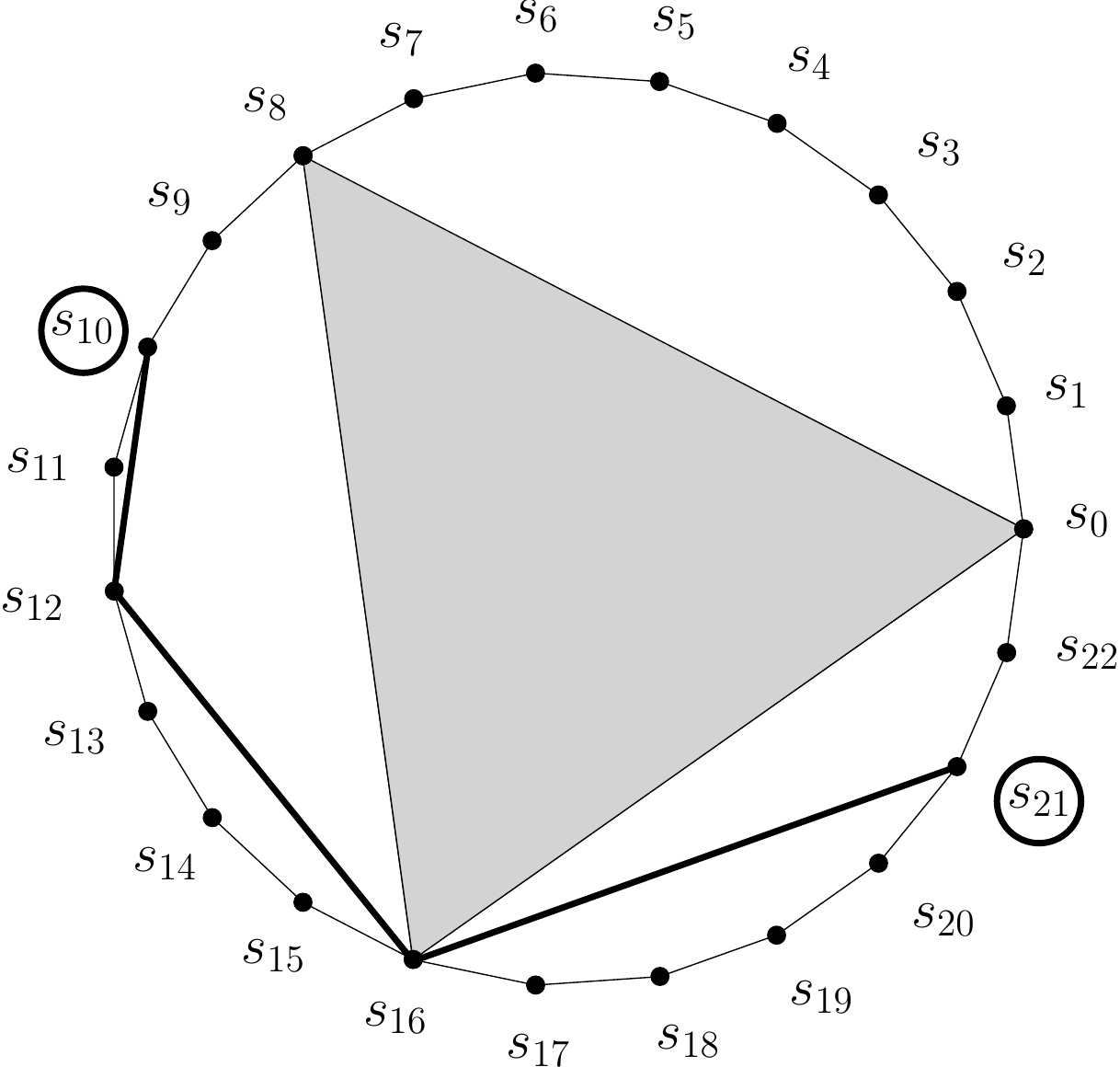}
\hspace{15mm}
\includegraphics[scale=0.5]{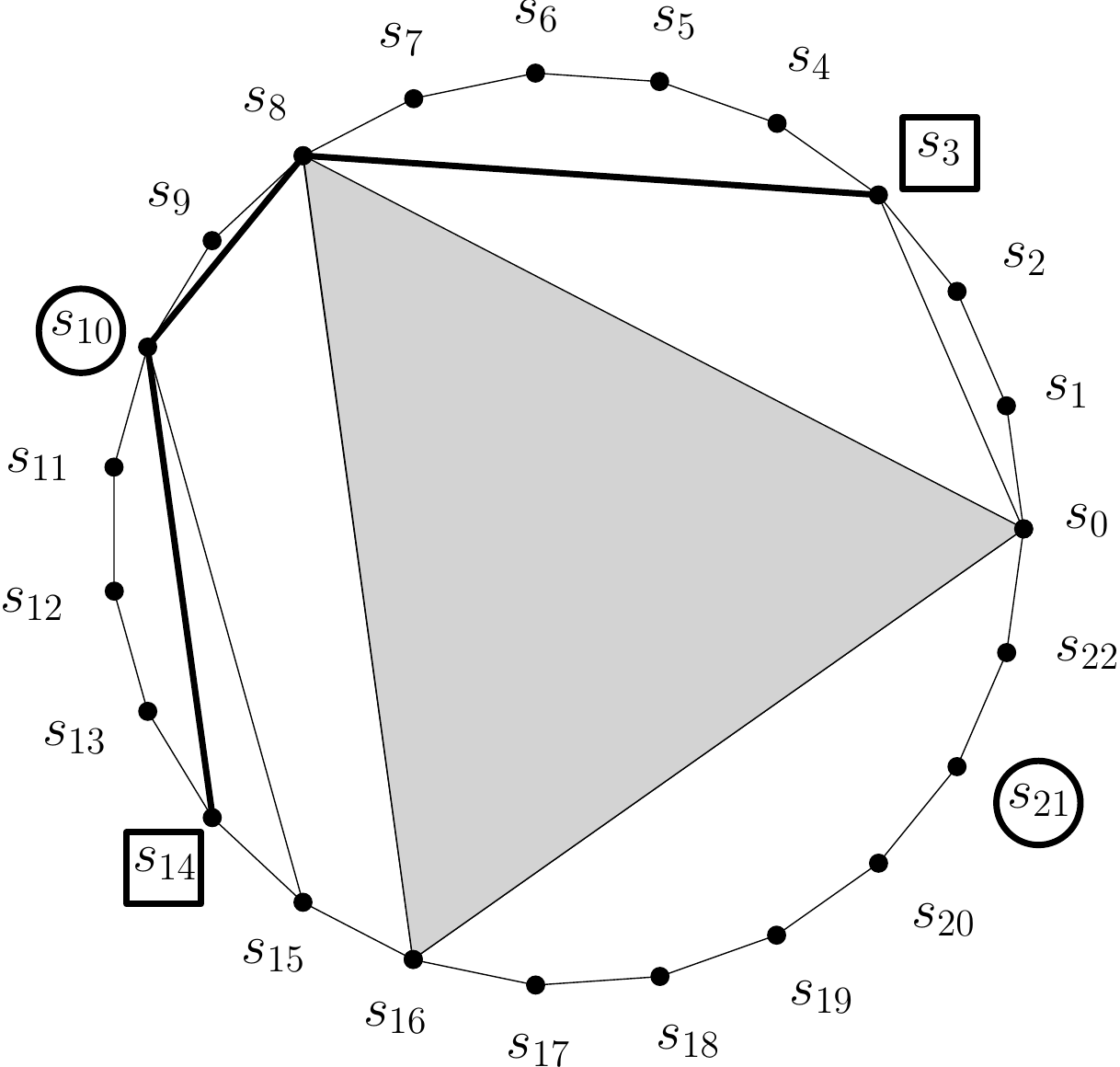}
\caption{Illustrating Lemma~\ref{lem:s_23.8}.
  Left: $s_{12} \in \pi(s_{10},s_{21})$, primary pair: $s_{10},s_{21}$.
  Right: $s_{10}s_{15} \in \pi(s_{10},s_{21})$, primary pair: $s_{10},s_{21}$,
  secondary pair: $s_3,s_{14}$.} 
\label{fig:s23_8}
\end{figure}

Refer to Fig.~\ref{fig:s23_8}. Let $s_0s_8$ be the longest
chord. The triangle with base $s_0s_8$ and third vertex in
$\lo(s_0s_8)$ has two other sides of convex hull lengths $7$ and $8$.  
It thus suffices to consider the triangle $\Delta{s_0 s_8 s_{16}}$ only and
assume that the edges $s_0s_8$, $s_8 s_{16}$ and $s_{0} s_{16}$ are present. 

In this proof, the primary pair is $s_{10},s_{21}$ and the secondary pair is $s_3,s_{14}$.
Now, consider the pair $s_{10},s_{21}$. Note that either $s_{0} \in
 \pi(s_{10},s_{21})$ or $s_{16} \in  \pi(s_{10},s_{21})$. In the
 former case, $\delta(s_{10},s_{21}) \geq
 |\rho(10,8,0,21)| / |s_{10}s_{21}| \geq f(2,8,2) = 1.4308\ldots$
 We may thus assume that $s_{16} \in \pi(s_{10}, s_{21})$.
 
 Similarly, for the pair $s_3,s_{14}$ either $s_0 \in \pi(s_3,s_{14})$
 or $s_8 \in \pi(s_3,s_{14})$. If $s_0 \in \pi(s_3,s_{14})$, then
 $\delta(s_3,s_{14}) \geq |\rho(3,0,16,14)| / |s_{3}s_{14}| \geq f(3,7,2) = 1.4886\ldots$
 Thus, assume that $s_{8} \in \pi(s_{3}, s_{14})$.

If at least one of $s_{12}, s_{13}$, or $s_{14}$ is in $\pi(s_{10},s_{21})$, then 
\begin{align*}
  \delta(s_{10},s_{21}) &\geq
  \frac{\min(|\rho(10,12,16,21)|, |\rho(10,13,16,21)|, |\rho(10,14,16,21)|)}{|s_{10}s_{21}|} \\
&\geq \min(f(2,4,5),f(3,3,5),f(4,2,5)) = f(2,4,5) = 1.4237\ldots
\end{align*}

Otherwise, one of $s_{10}s_{15},s_{10}s_{16}, s_{11}s_{15}$, or $s_{11}s_{16}$
must be in $\pi(s_{10},s_{21})$, and
\begin{align*}
  \delta(s_3,s_{14}) &\geq
  \frac{\min(|\rho(3,8,10,14)|, |\rho(3,8,11,14)|),|\rho(3,8,15,14)|) }{|s_3s_{14}|}\\ 
&\geq \min(f(5,2,4),f(5,3,3),f(1,5,7)) =  f(2,4,5) = 1.4237\ldots \qedhere
\end{align*} 
\end{proof}

\begin{lemma}{\label{lem:s_23.9}}
If $\ell=9$, then $\delta(T) \geq f(2,4,5)= 1.4237\ldots$
\end{lemma}
\begin{proof} Let $s_0s_9$ be the longest chord and consider the
  triangle with base $s_0s_9$ and the third vertex in $\lo(s_0s_9)$.
  There are three possible cases depending on the convex
  hull lengths of other two sides of the triangle: $\{7,7\}$,
  $\{8,6\}$ or $\{9,5\}$. We consider them successively. 

\begin{figure}[hbtp]
	\centering
	\includegraphics[scale=0.5]{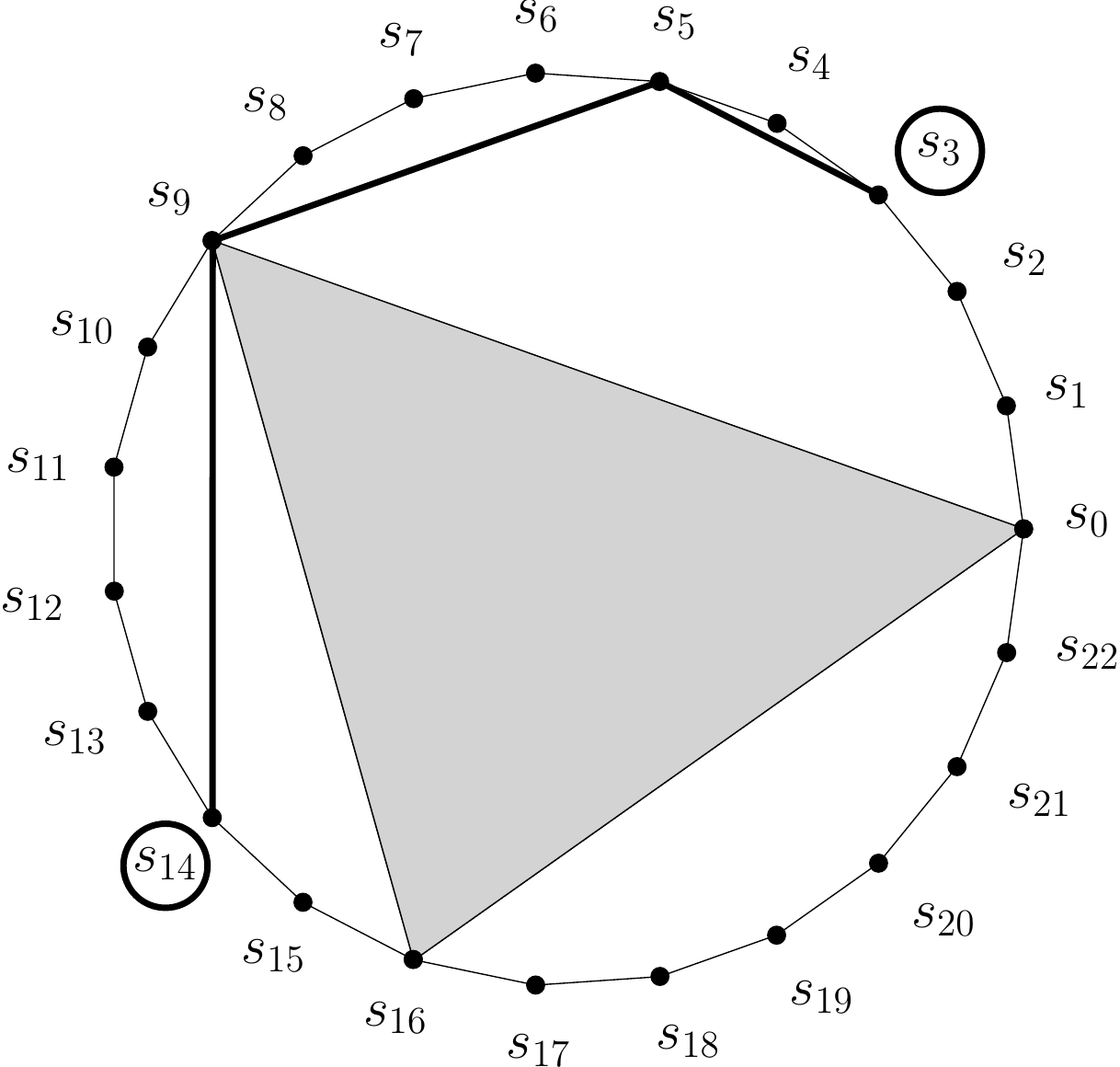}
	\hspace{15mm}
	\includegraphics[scale=0.5]{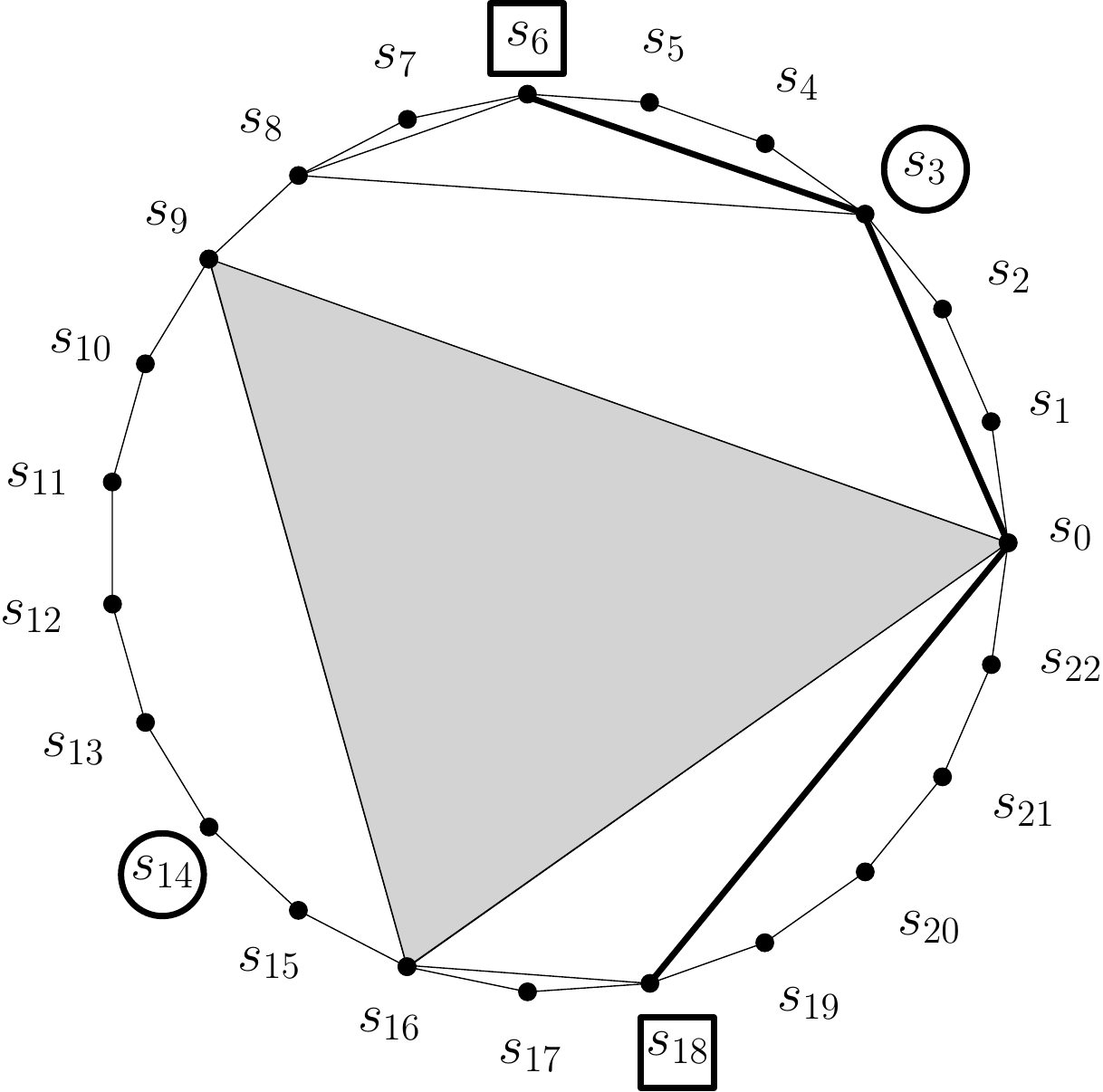}
	\caption{Illustrating \textsc{Case A} from Lemma~\ref{lem:s_23.9}.
          Left: $s_5 \in \pi(s_3,s_{14})$. Right: $s_3s_8 \in \pi(s_3,s_{14})$. } 
	\label{fig:s23_9A}
\end{figure}

\textsc{Case A}: The convex hull lengths of the other two sides are $\{7,7\}$.
Let $\Delta s_0s_9s_{16}$ be the required triangle; refer to Fig.~\ref{fig:s23_9A}.
In this case, the primary pair is 
$s_{3},s_{14}$ and the secondary pair is $s_6,s_{18}$.
Either $s_0 \in \pi(s_{3},s_{14})$ or $s_{9} \in \pi(s_{3},s_{14})$.
If $s_0 \in \pi(s_{3},s_{14})$, then
$\delta(s_{3},s_{14}) \geq |\rho(3,0,16,14)| / |s_{3}s_{14}| \geq 
f(3,7,2) $ $= 1.4886\ldots$ Thus, we assume that $s_{9} \in\pi(s_{3},s_{14})$. 

Similarly, for the pair $s_6,s_{18}$, either $s_9 \in \pi(s_6,s_{18})$ or
$s_0 \in \pi(s_6,s_{18})$. If $s_9 \in \pi(s_6,s_{18})$
then $\delta(s_6,s_{18})  \geq |\rho(6,9,16,18)| / |s_6s_{18}|\geq f(3,7,2) = 1.4886\ldots$
Thus, assume that $s_0 \in \pi(s_6,s_{18})$.

Now, if at least one of $s_5, s_6$, or $s_7$ is in $\pi(s_3,s_{14})$, then 
\begin{align*}
  \delta(s_3,s_{14})
  &\geq \frac{\min(|\rho(3,5,9,14)|, |\rho(3,6,9,14)|, |\rho(3,7,9,14)|)}{|s_3s_{14}|} \\
&\geq \min(f(2,4,5),f(3,3,5),f(4,2,5)) = f(2,4,5) = 1.4237\ldots
\end{align*}

Otherwise, one of $s_3s_8, s_3s_9, s_4s_8$, or $s_4s_9$ must be in $\pi(s_3,s_{14})$, and
\begin{align*}
\delta(s_6,s_{18}) &\geq \frac{\min(|\rho(6,3,0,18)|, |\rho(6,4,0,18)|)}{|s_6s_{18}|} 
\geq \min(f(3,3,5),f(2,4,5)) =  f(2,4,5) = 1.4237\ldots
\end{align*} 

\textsc{Case B}: The convex hull lengths of the other two
sides are $\{8,6\}$. Let $\Delta s_0s_9s_{17}$ be the required
triangle; refer to Fig.~\ref{fig:s23_9B-C}\,(left). As in \textsc{Case A}, the primary pair is 
$s_{3},s_{14}$ and the secondary pair is $s_6,s_{18}$. Consider the pair $s_{3},s_{14}$.
Either $s_{17} \in \pi(s_{3},s_{14})$ or $s_{9} \in \pi(s_{3},s_{14})$.
If $s_{17} \in \pi(s_{3},s_{14})$, then
$\delta(s_{3},s_{14}) \geq |\rho(3,0,17,14)| / |s_3s_{14}| \geq f(3,6,3) = 1.5312\ldots$
So we assume that $s_{9} \in \pi(s_{3},s_{14})$.
\begin{figure}[hbtp]
	\centering
	\includegraphics[scale=0.5]{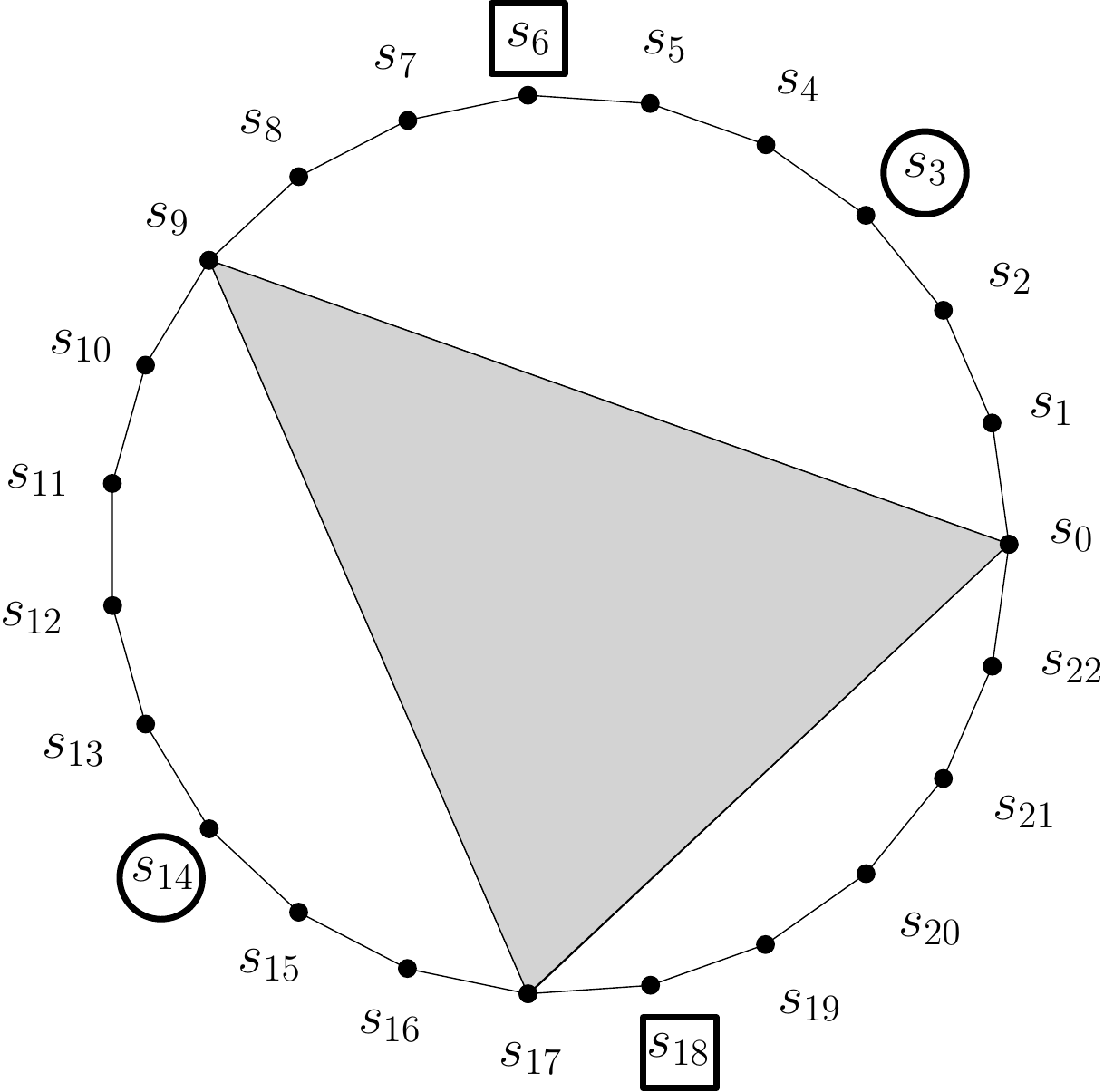}
	\hspace{15mm}
	\includegraphics[scale=0.5]{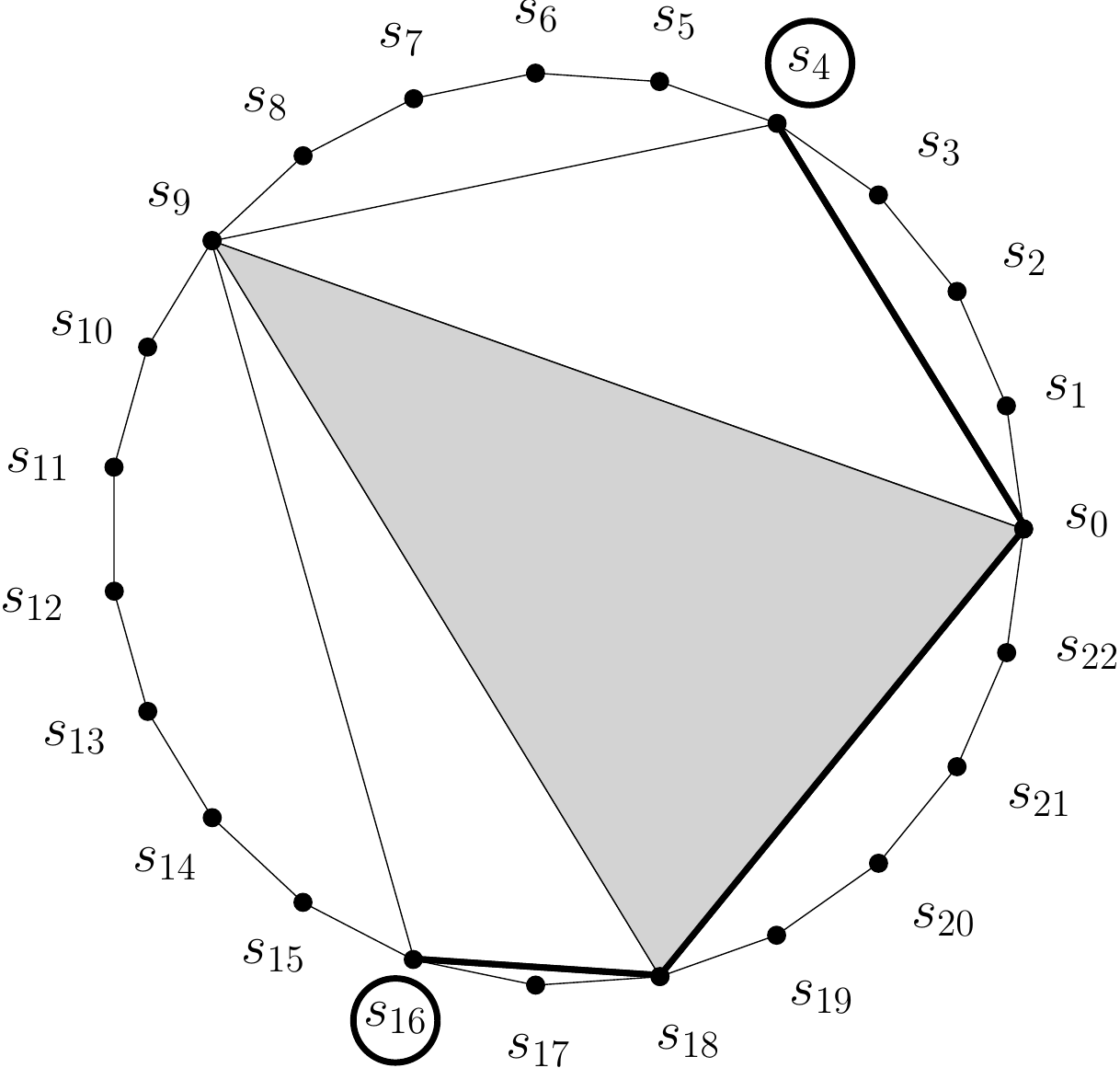}
	\caption{Illustrating \textsc{Case B}\,(left) and \textsc{Case C}\,(right)
          from Lemma~\ref{lem:s_23.9}. }
	\label{fig:s23_9B-C}
\end{figure}

Similarly, for the pair $s_6,s_{18}$, either $s_9 \in \pi(s_6,s_{18})$ or $s_0 \in \pi(s_6,s_{18})$.
If $s_9 \in \pi(s_6,s_{18})$ then
$\delta(s_6,s_{18})  \geq |\rho(6,9,17,18)| / |s_6s_{18}|\geq f(3,8,1) = 1.4257\ldots$
Thus, we assume that $s_0 \in \pi(s_6,s_{18})$. Now, it can be checked that by the same analysis
as in \textsc{Case A}, the same lower bound of $f(2,4,5)$ holds.

\smallskip
\textsc{Case C}: The convex hull lengths of the other two
sides of the triangle are $\{9,5\}$. Let $\Delta{s_0 s_9 s_{18}}$ be the required triangle;
refer to Fig.~\ref{fig:s23_9B-C}\,(right). Then,
\begin{equation*}
\delta(s_{4},s_{16}) \geq \frac{\min(|\rho(4,0,18,16)|,|\rho(4,9,16)|)}{|s_{4}s_{16}|} 
\geq \min(f(4,5,2),f(5,7)) = f(2,4,5)= 1.4237\ldots
\qedhere
\end{equation*}
\end{proof}

\begin{lemma}{\label{lem:s_23.10}}
If $\ell=10$, then $\delta(T) \geq f(2,4,5)= 1.4237\ldots$
\end{lemma}
\begin{proof} Let $s_0s_{10}$ be the longest chord. The possible
  convex hull lengths of the other two sides of the triangle with base
  $s_0s_{10}$ and the third vertex in $\lo(s_0s_{10})$ are
  $\{10,3\}$,$\{9,4\}$,$\{8,5\}$,$\{7,6\}$. We consider these cases
  successively. 
\begin{figure}[hbtp]
	\centering
	\includegraphics[scale=0.5]{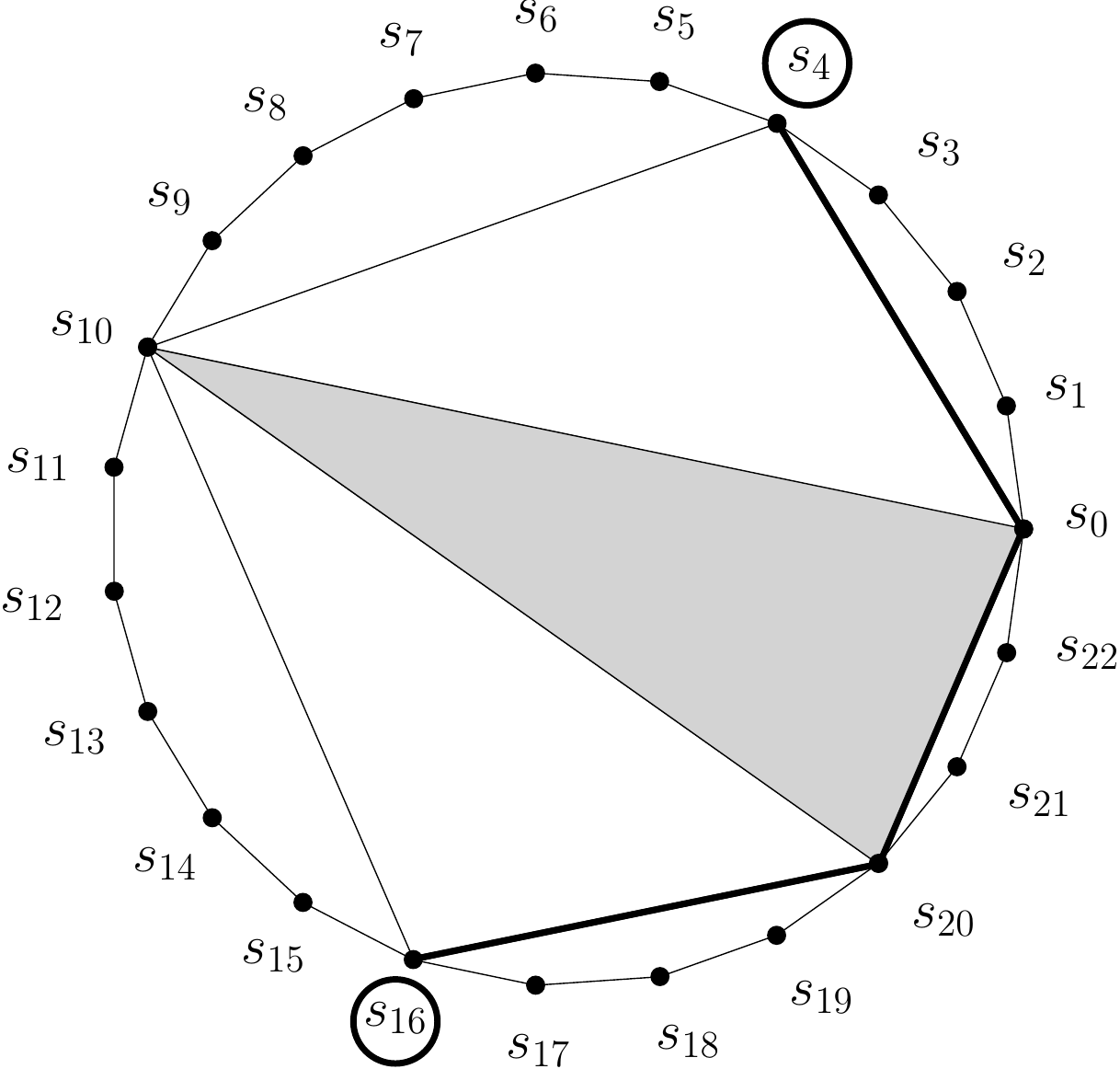}
	\hspace{15mm}
	\includegraphics[scale=0.5]{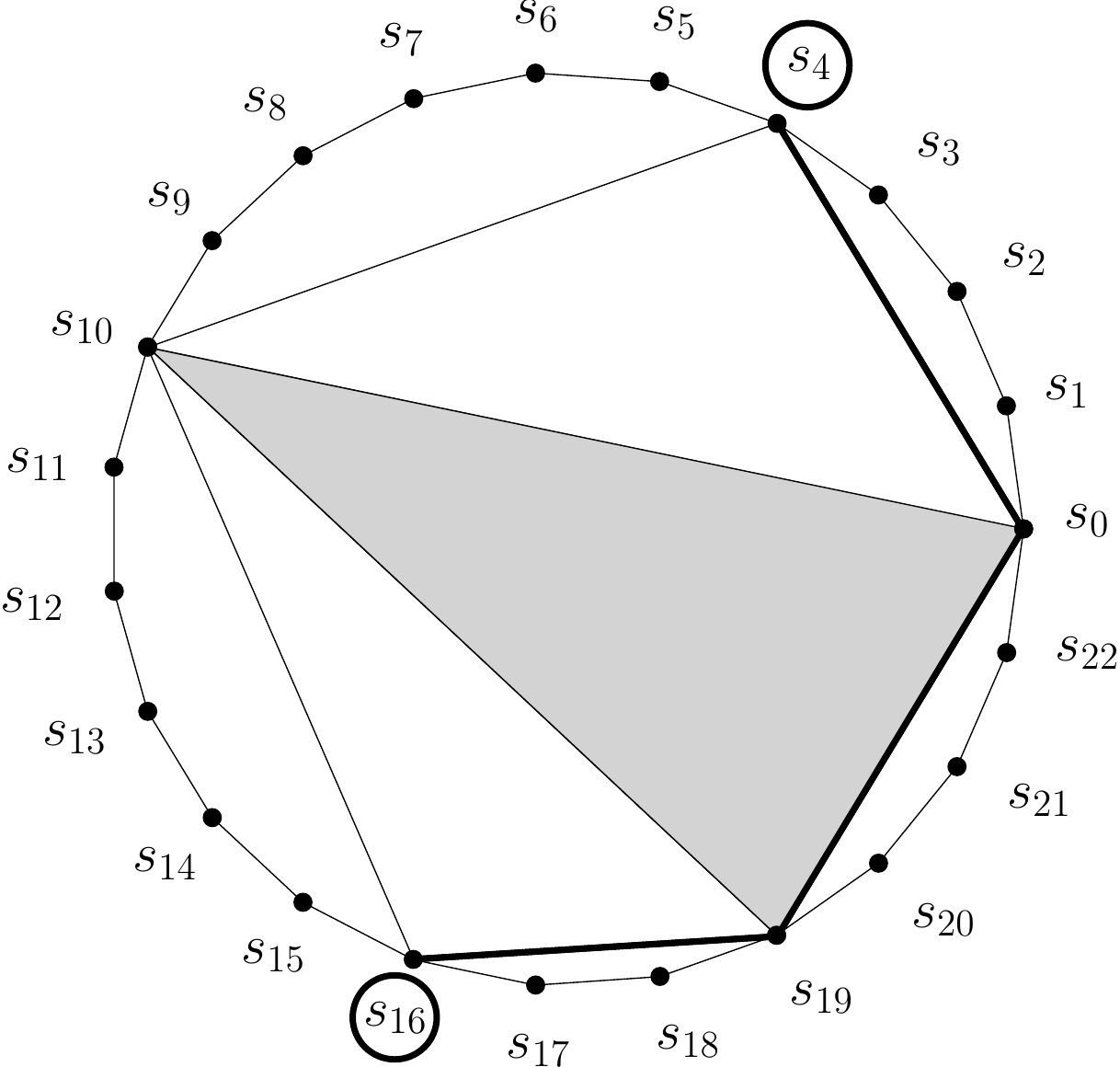}
	\caption{Illustrating \textsc{Case A} from Lemma~\ref{lem:s_23.10}. }
	\label{fig:s23_10A-B}
\end{figure}
	
\textsc{Case A}: The convex hull lengths of the other two
sides of the triangle are $\{10,3\}$, $\{9,4\}$ or $\{8,5\}$. 
Let $\Delta s_0s_{10}s_{20}$, $\Delta s_0s_{10}s_{19}$,  $\Delta s_0s_{10}s_{18}$ 
be the required triangles, respectively;
refer to Fig.~\ref{fig:s23_10A-B} and Fig.~\ref{fig:s23_10C-D}\,(left). Then, 
\begin{align*}
\delta(s_{4},s_{16}) &\geq
\frac{\min(|\rho(4,0,20,16)|,|\rho(4,0,19,16)|,|\rho(4,0,18,16)|,|\rho(4,10,16)|)}{|s_{4}s_{16}|}\\
&\geq \min(f(4,3,4),f(4,4,3),f(4,5,2),f(6,6)) = f(2,4,5)= 1.4237\ldots
\end{align*}

\begin{figure}[hbtp]
	\centering
	\includegraphics[scale=0.5]{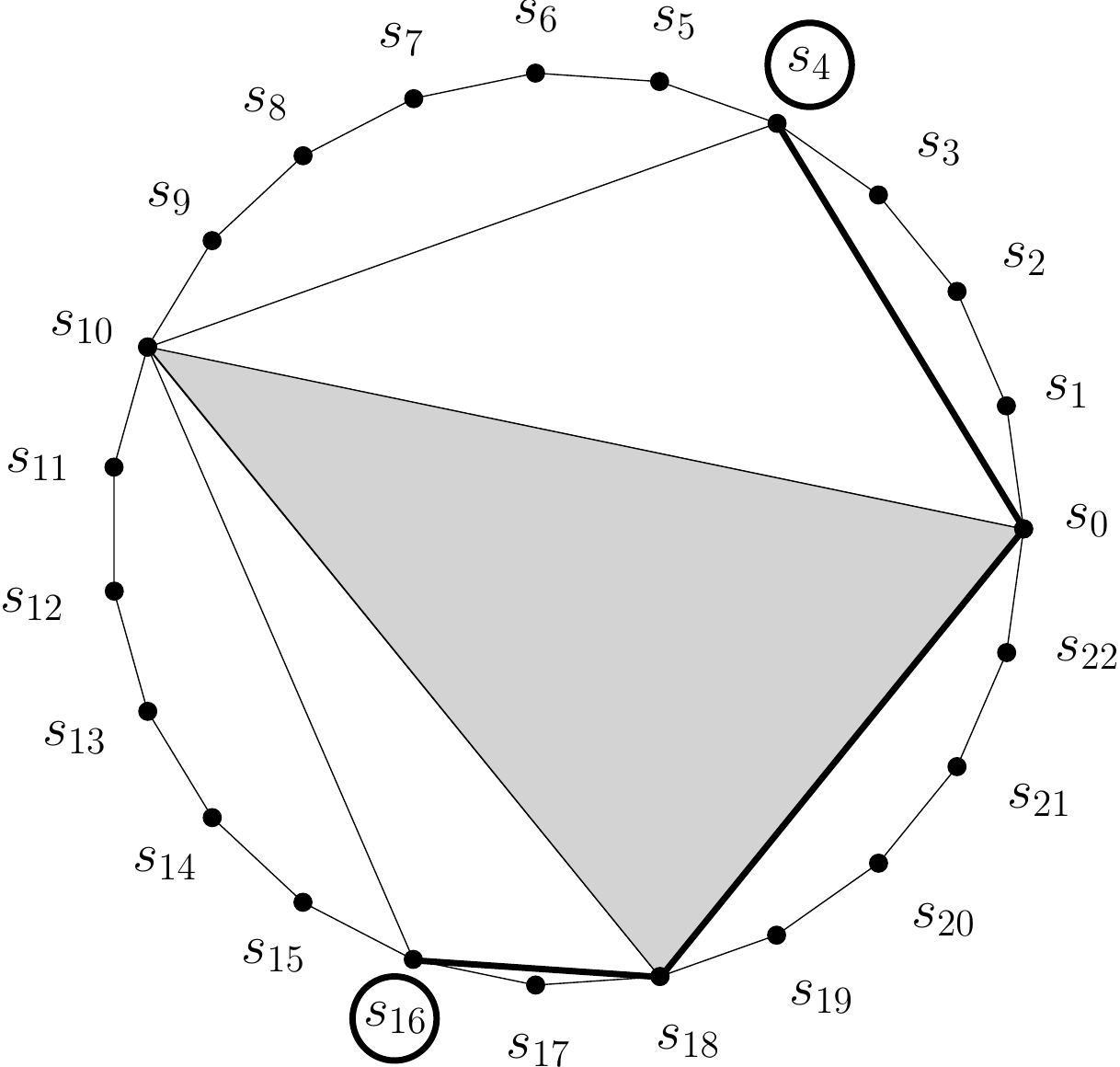}
	\hspace{15mm}
	\includegraphics[scale=0.5]{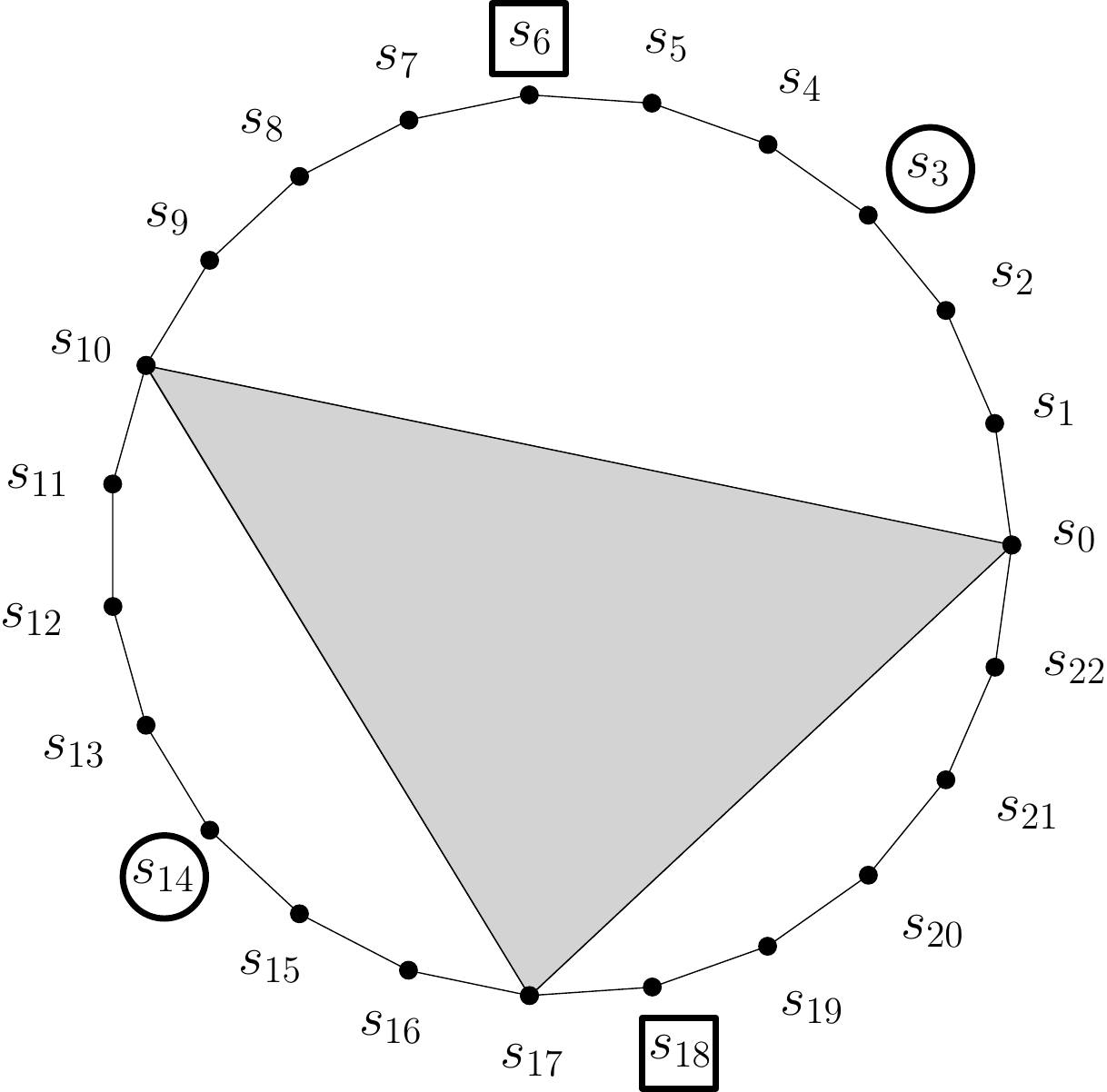}
	\caption{Illustrating \textsc{Case A}\,(left) and \textsc{Case B}\,(right)
		from Lemma~\ref{lem:s_23.10}. }
	\label{fig:s23_10C-D}
\end{figure}

\textsc{Case B}:  The convex
hull lengths of the other two sides are $\{7,6\}$. Let
$\Delta s_0s_{10}s_{17}$ be the required triangle; refer to Fig.~\ref{fig:s23_10C-D}\,(right).
In this case, the primary pair is 
$s_{3},s_{14}$ and the secondary pair is $s_6,s_{18}$.
Either $s_{0} \in \pi(s_{3},s_{14})$ or $s_{10} \in \pi(s_{3},s_{14})$.
If $s_{0} \in \pi(s_{3},s_{14})$, then
$\delta(s_{3},s_{14}) \geq |\rho(3,0,17,14)| / |s_{3}s_{14}| \geq 
f(3,6,3) = 1.5312\ldots$ Thus, we assume that $s_{10} \in\pi(s_{3},s_{14})$.

Similarly, for the pair $s_6,s_{18}$, either $s_{10} \in \pi(s_6,s_{18})$
or $s_0 \in \pi(s_6,s_{18})$. If $s_{10} \in \pi(s_6,s_{18})$, then
$\delta(s_6,s_{18})  \geq |\rho(6,10,17,18)| / |s_6s_{18}|\geq f(4,7,1) = 1.4761\ldots$
Thus, assume that $s_0 \in \pi(s_6,s_{18})$.

Now, if at least one of $s_5, s_6$, $s_7$, or $s_8$ is in $\pi(s_3,s_{14})$, then 
\begin{align*}
  \delta(s_3,s_{14}) &\geq
  \frac{\min(|\rho(3,5,10,14)|, |\rho(3,6,10,14)|, |\rho(3,7,10,14)|,|\rho(3,8,10,14)|)}{|s_3s_{14}|} \\
  &\geq \min(f(2,5,4),f(3,4,4),f(4,3,4),f(5,2,4)) = f(2,4,5) = 1.4237\ldots
\end{align*}

Otherwise, one of $s_3s_9, s_3s_{10}, s_4s_9$, or $s_4s_{10}$ must be in $\pi(s_3,s_{14})$, and
\begin{align*}
	\delta(s_6,s_{18}) &\geq \frac{\min(|\rho(6,3,0,18)|, |\rho(6,4,0,18)|)}{|s_6s_{18}|} \\
	&\geq \min(f(3,3,5),f(2,4,5)) =  f(2,4,5) = 1.4237\ldots \qedhere
\end{align*} 
\end{proof}

\begin{lemma}{\label{lem:s_23.11}}
If $\ell=11$, then $\delta(T) \geq f(2,4,5)= 1.4237\ldots$
\end{lemma}
\begin{proof} Let $s_0s_{11}$ be the longest chord. Since the size of
  $\up(s_0s_{11})$ is smaller than the size of  $\lo(s_0s_{11})$, we consider
  $\up(s_0s_{11})$ is our analysis. The possible
  convex hull lengths of the other two sides of the triangle with
  base $s_0s_{11}$ and the third vertex in $\up(s_0s_{11})$ are
  $\{1,10\}$, $\{2,9\}$, $\{3,8\}$, $\{4,7\}$, $\{5,6\}$. We consider the
  following cases successively. 

\begin{figure}[hbtp]
	\centering
	\includegraphics[scale=0.5]{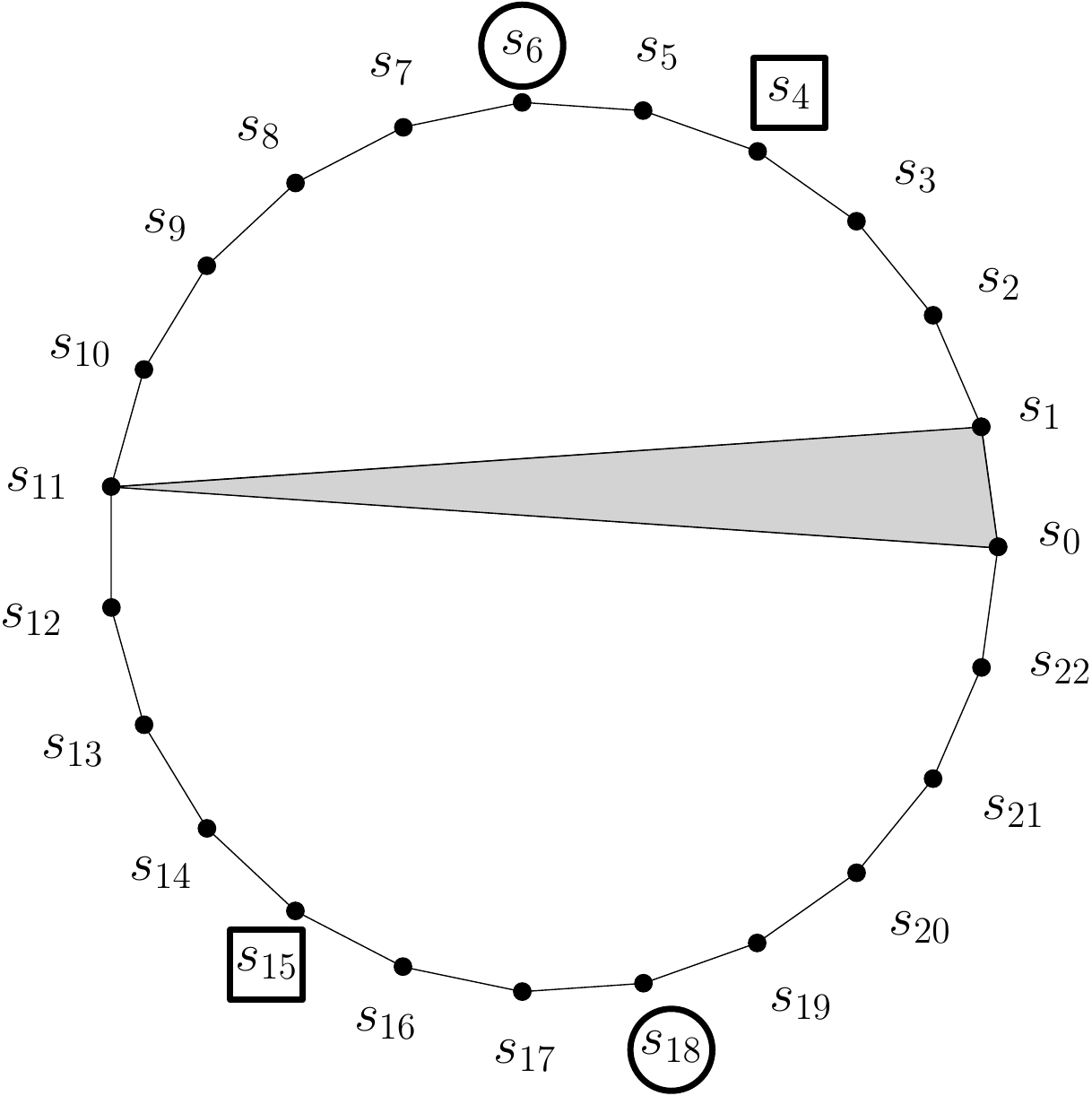}
	\hspace{15mm}
	\includegraphics[scale=0.5]{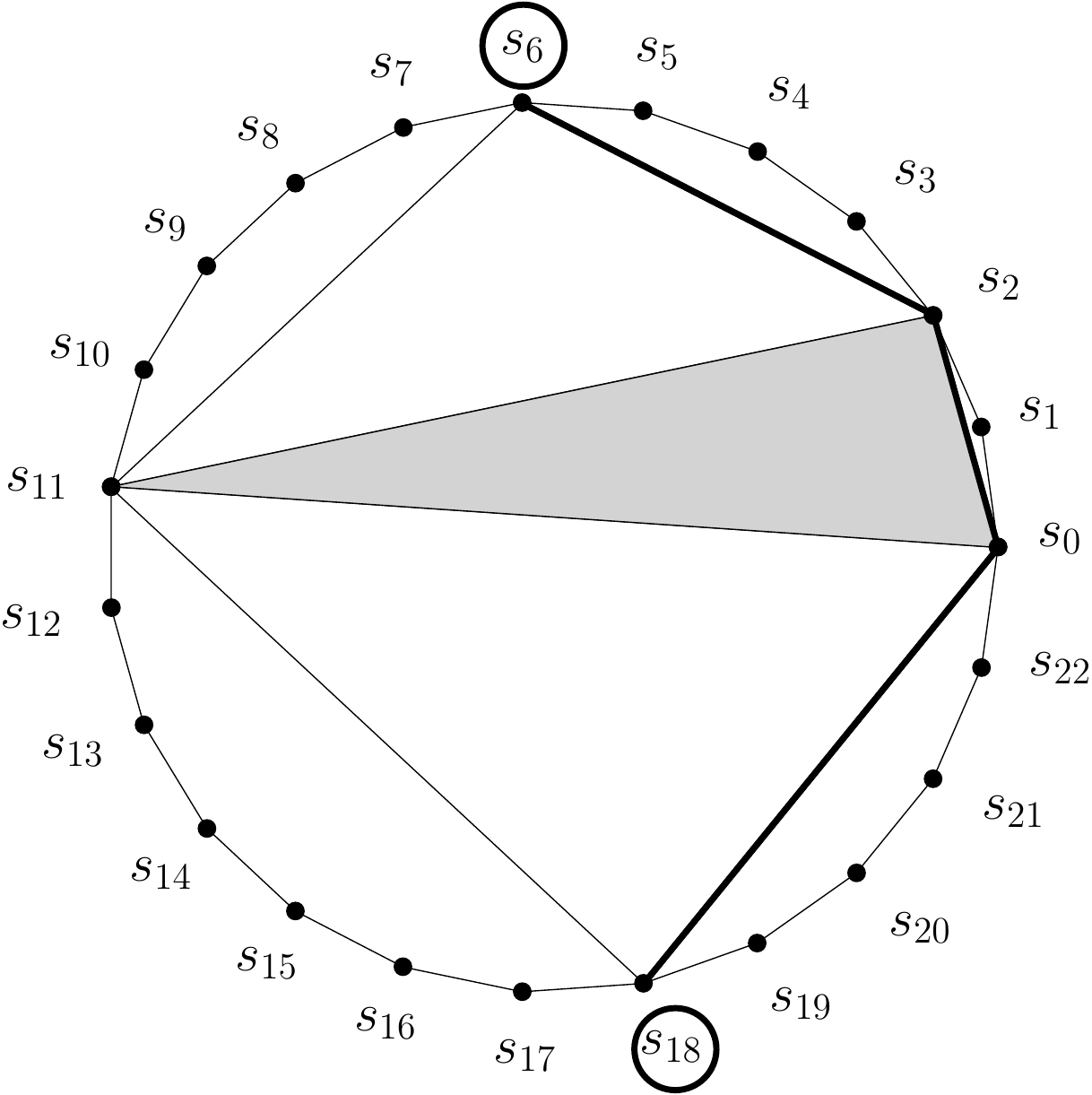}
	\caption{Illustrating \textsc{Case A}\,(left) and \textsc{Case B}\,(right) 
          from Lemma~\ref{lem:s_23.11}. }
	\label{fig:s23_11A1-B}
\end{figure}
	
\textsc{Case A}: The convex
hull lengths of the other two sides are $\{1,10\}$. Let
$\Delta s_0s_1s_{11}$ be the required triangle; refer to Fig.~\ref{fig:s23_11A1-B}\,(left).
In this case, the primary pair is 
$s_6,s_{18}$ and the secondary pair is $s_4,s_{15}$. Consider the pair
$s_{6},s_{18}$. Either $s_{11} \in \pi(s_{6},s_{18})$ or $s_{0} \in \pi(s_{6},s_{18})$.
If $s_{11} \in \pi(s_{6},s_{18})$, then
$\delta(s_{6},s_{18}) \geq |\rho(6,11,18)|  /  |s_{6}s_{18}| \geq
f(5,7) = 1.4514\ldots$ Hence, we assume that $s_{0} \in
\pi(s_{6},s_{18})$. 
\smallskip

If at least one of $s_2, s_3, s_4$, or  $s_5$ is in $\pi(s_6,s_{18})$, then
\begin{align*}
\delta(s_6,s_{18}) &\geq
\frac{\min(|\rho(6,2,1,0,18)|,|\rho(6,3,1,0,18)|,|\rho(6,4,1,0,18)|,|\rho(6,5,1,0,18)|)}{|s_6s_{18}|}\\
&= \min(f(4,1,1,5),f(3,2,1,5),f(2,3,1,5),f(1,4,1,5)) = f(1,1,4,5) =  1.4263\ldots
\end{align*}
Otherwise,  $s_1s_6$ is in $\pi(s_6,s_{18})$, and then 
$$\delta(s_{4},s_{15}) \geq
\frac{\min(|\rho(4,6,11,15)|,|\rho(4,1,0,15)|)}{|s_{4}s_{15}|} \geq \min(f(2,5,4),f(3,1,8)) =
f(2,4,5) = 1.4237\ldots$$ 

\textsc{Case B}: The convex hull lengths of the other two sides are $\{2,9\}, \{3,8\}$ or $\{4,7\}$.
Let $\Delta s_0s_2s_{11}$, $\Delta s_0s_3s_{11}$, $\Delta s_0s_4s_{11}$ be the required triangles,
respectively. Refer to Fig.~\ref{fig:s23_11A1-B}\,(right) and Fig.~\ref{fig:s23_11C-D}\,(left)
for illustrations.

\begin{figure}[hbtp]
	\centering
	\includegraphics[scale=0.5]{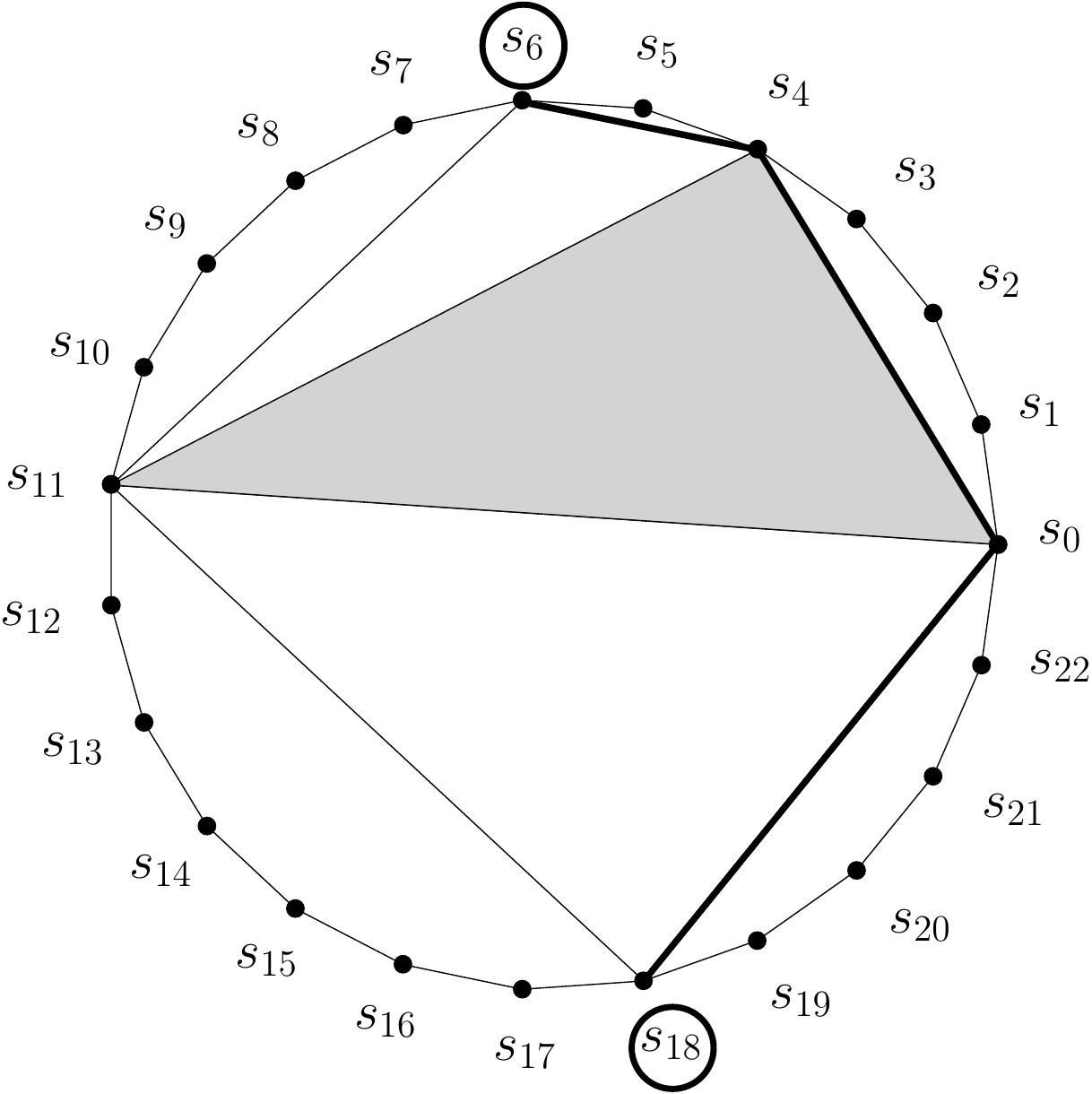}
	\hspace{15mm}
	\includegraphics[scale=0.5]{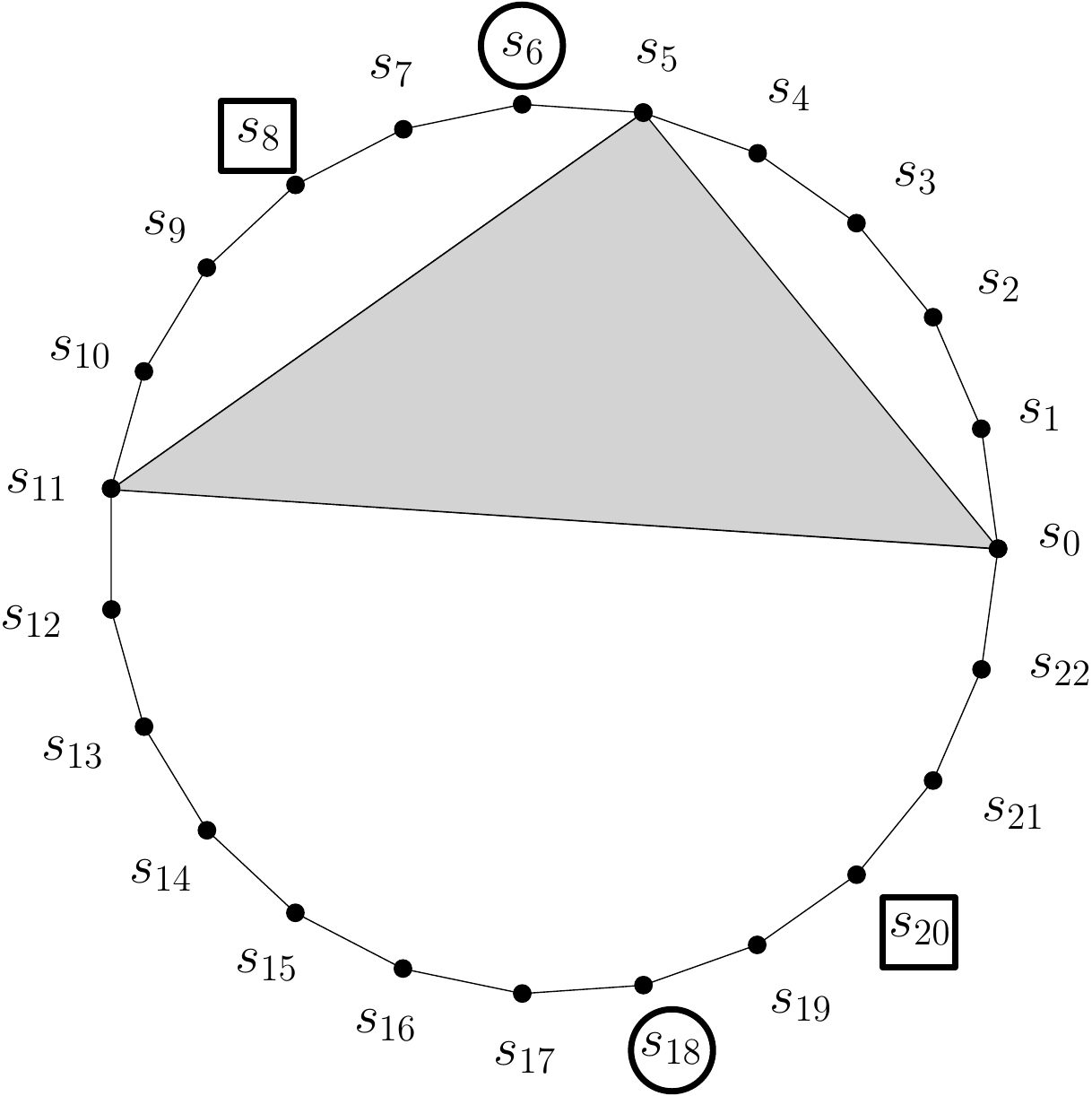}
	\caption{Illustrating \textsc{Case B}\,(left)
          and \textsc{Case C}\,(right) from Lemma~\ref{lem:s_23.11}.}
	\label{fig:s23_11C-D}
\end{figure}

As in the \textsc{Case A}, we may assume that $s_0 \in \pi(s_{6},s_{18})$. Then,
\begin{align*}
  \delta(s_{6},s_{18}) &\geq
  \frac{\min(|\rho(6,2,0,18)|,|\rho(6,3,0,18)|,|\rho(6,4,0,18)| )}{|s_{6}s_{18}|}\\
	&\geq \min(f(4,2,5),f(3,3,5),f(2,4,5)) =  f(2,4,5)= 1.4237\ldots
\end{align*}

\textsc{Case C}: The convex hull lengths of the other two sides are $\{5,6\}$.
Let $\Delta s_0s_5s_{11}$ be the required triangle; refer to Fig.~\ref{fig:s23_11C-D}\,(right).
In this case, the primary pair is $s_6,s_{18}$ and the secondary pair is $s_8,s_{20}$. As
in the \textsc{Case A}, we assume that $s_0 \in \pi(s_{6},s_{18})$. 

Now, if at least one of $s_{19}, s_{20}, s_{21}$, or  $s_{22}$ is in $\pi(s_6,s_{18})$, then
\begin{align*}
\delta(s_6,s_{18}) &\geq
\frac{\min(|\rho(6,5,0,19,18)|,|\rho(6,5,0,20,18)|,|\rho(6,5,0,21,18)|,
  |\rho(6,5,0,22,18)|)}{|s_6s_{18}|}\\
&= \min(f(1,5,4,1),f(1,5,3,2),f(1,5,2,3),f(1,5,1,4)) = f(1,1,4,5) =  1.4263\ldots
\end{align*}

Otherwise, $s_0s_{18}$ is in $\pi(s_6,s_{18})$, and then 
\begin{align*}
\delta(s_{8},s_{20}) &\geq  \frac{\min
	(|\rho(8,5,0,20)|,|\rho(8,11,18,20)|)}{|s_{8}s_{20}|} \\
&\geq \min(f(3,5,3),f(3,7,2)) = f(3,3,5) = 1.4312\ldots \qedhere
\end{align*}
\end{proof}

Putting these facts together yields the main result of this section:
\begin{theorem} \label{thm-s23}
Let $S$ be a set of $23$ points placed at the vertices of a regular $23$-gon. Then
$$ \delta_0(S) = f(2,2,8) =
\left(2 \sin \frac{2 \pi}{23} + \sin \frac{8 \pi}{23}\right)
\bigg/ \sin\frac{11 \pi}{23} = 1.4308\ldots $$
\end{theorem}
\begin{proof}
  By Lemmas~\ref{lem:s_23.8}-\ref{lem:s_23.11},
  we conclude that  $\delta_0(S) \geq f(2,4,5)=
  (\sin\frac{2\pi}{23}+\sin\frac{4\pi}{23}+\sin\frac{5\pi}{23})/\sin\frac{11\pi}{23}
  =1.4237\ldots$ On the other hand,
the triangulation of $S$ in Fig.~\ref{fig:s23}\,(right) has stretch factor
$f(2,2,8) =1.4308\ldots$
and thus $f(2,4,5)=1.4237\ldots \leq \delta_0(S) \leq f(2,2,8) = 1.4308\ldots$

A parallel C\texttt{++} program\footnote{Refer to the ${\tt .cpp}$ file
  and the Appendix within the source at \url{arXiv:1509.07181}.}
that generates all triangulations of $S$ based on a low memory algorithm
by Parvez~\etal~\cite[Section 4]{PRN11}
shows that each of the $C_{21}$ triangulations has stretch factor at least $f(2,2,8)$.
We thereby obtain the following final result: $\delta_0(S) = f(2,2,8) =
(2\sin\frac{2\pi}{23}+\sin\frac{8\pi}{23})/\sin\frac{11\pi}{23}=1.4308\ldots$
\end{proof}

\paragraph{Remarks.} 
Using the program we have also checked that the next largest stretch factor
among all triangulations is $f(3,3,5)= 1.4312\ldots$, and further that 
there is no triangulation of $S$ that has stretch-factor $<1.4312$ other than
$f(2,2,8)$. Thus, the result in Theorem~\ref{thm-s23} is not affected by floating-point
precision errors.
\begin{figure}[htpb]
	\centering 
	\includegraphics[scale=0.52]{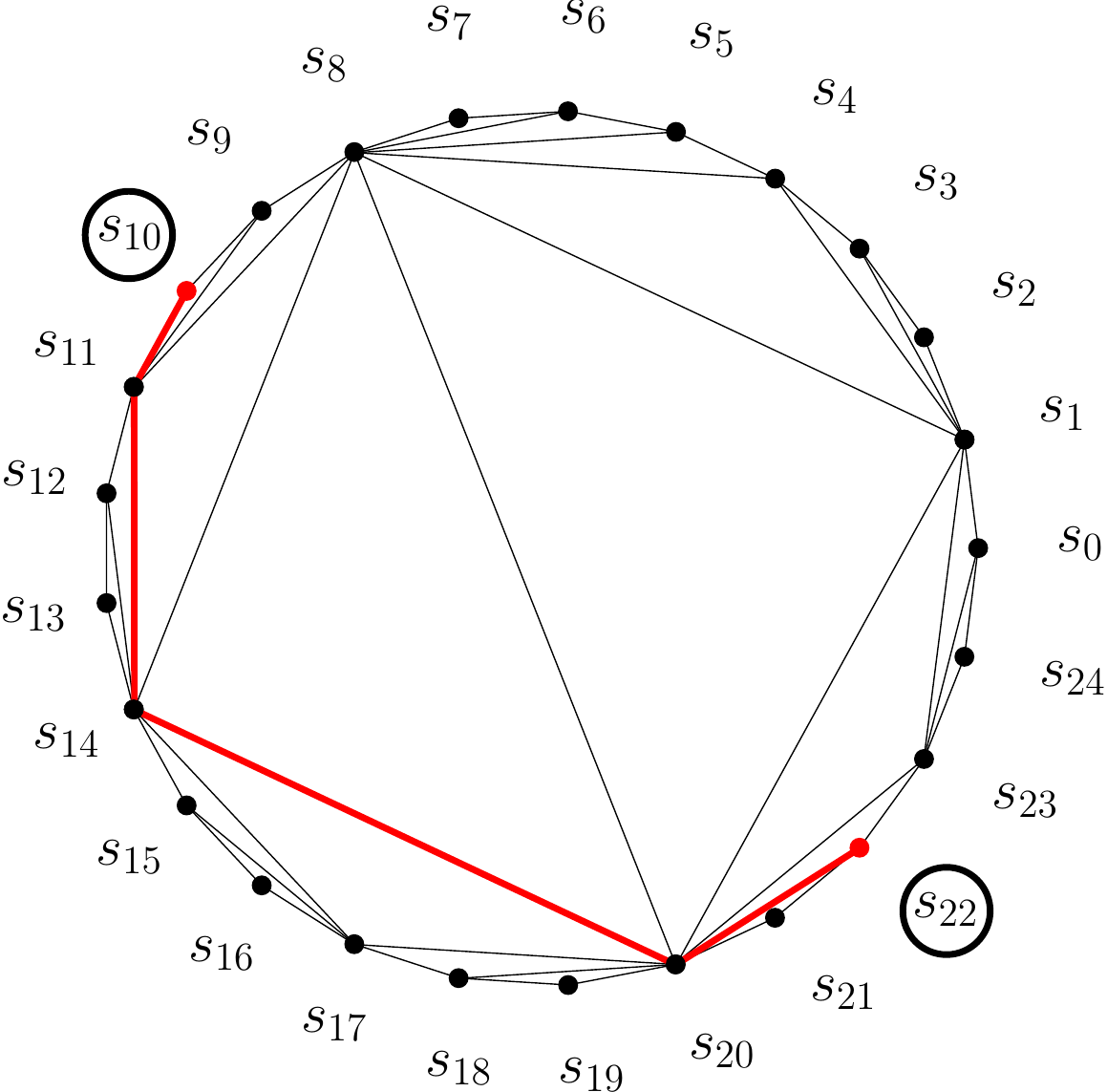}
	\hspace{15mm}
	\includegraphics[scale=0.52]{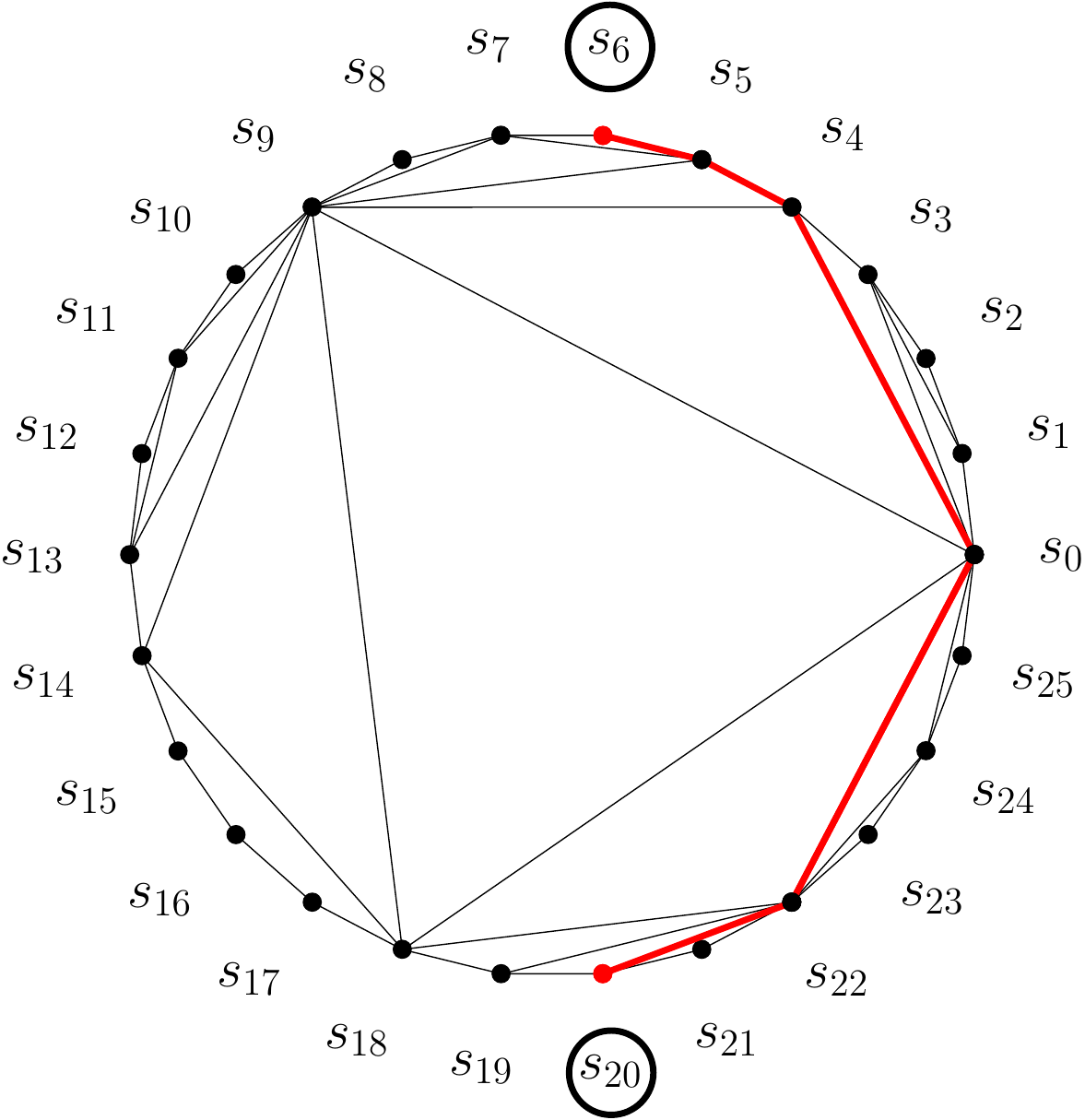}
	\caption{Triangulations of $S_{25}$ and $S_{26}$ with stretch factors 
	  $< 1.4296$ and $<1.4202$, respectively. Worst stretch factor pairs are marked
          in circles and the corresponding shortest paths are shown in red.}
	\label{fig:table}
\end{figure}

Let $S_n$ denote the set of points placed at the vertices of a regular $n$-gon.
Using a computer program, Mulzer obtained the values $\delta_0(S_n)$
for $4 \leq n \leq 21$ in his thesis~\cite[Chapter 3]{Mu04}. 
Using our C\texttt{++} program, we confirmed the previous values
and extended the range up to $n=24$: $\delta_0(S_{22}) =1.4047\ldots$,
$\delta_0(S_{24}) =  1.4013\ldots$ and somewhat surprisingly,
$\delta_0(S_{23}) =  1.4308\ldots$ By upper bound constructions,
it follows that $\delta_0(S_{25}) < 1.4296$ and $\delta_0(S_{26}) < 1.4202$;
see Fig.~\ref{fig:table}.
Observe that $\delta_0(S_n)$ does not exhibit a monotonic behavior;
see Table~\ref{table:Sn}. 

\begin{table}[htpb]
\begin{center}
\begin{tabular}{||c|c||c|c||c|c||}
	\hline
	$n$ & $\delta_0(S_n)$ & $n$ & $\delta_0(S_n)$ & $n$ & $\delta_0(S_n)$\\
	\hline
	4 & $1.4142\ldots$ & 12 & $1.3836\ldots$ & 20 & $1.4142\ldots$  \\
	5 & $1.2360\ldots$ & 13 & $1.3912\ldots$ & 21 & $1.4161\ldots$\\
	6 & $1.3660\ldots$ & 14 & $1.4053\ldots$ & 22 & $1.4047\ldots$\\
	7 & $1.3351\ldots$ & 15 & $1.4089\ldots$ & \textbf{23} & $\textbf{1.4308}\ldots$\\
	8 & $1.4142\ldots$ & 16 & $1.4092\ldots$ & 24 & $1.4013\ldots$\\
	9 & $1.3472\ldots$ & 17 & $1.4084\ldots$ & 25 & $< 1.4296$\\
        10 & $1.3968\ldots$ & 18 & $1.3816\ldots$& 26 & $< 1.4202$\\
	11 & $1.3770\ldots$ & 19 & $1.4098\ldots$& & \\
	\hline
\end{tabular}
\end{center}
\caption{The values of $\delta_0(S_n)$ for $n=4,\ldots,26$.}
\label{table:Sn}
\end{table}

\section{Lower bounds for the degree~$3$ and~$4$ dilation}
\label{sec:deg3and4}

In this section, we provide lower bounds for the worst case degree~$3$ and~$4$ dilation
of point sets in the Euclidean plane. We begin with degree~$3$ dilation. We first present
a set $P$ of $n=13$ points (a section of the hexagonal lattice with six boundary points removed)
that has $\delta_0(P,3) \geq 1+\sqrt{3}$ and then extend $P$ to achieve this lower bound
for any $n>13$. 

\begin{theorem} \label{thm:deg3}
For every $n\geq13$, there exists a set $S$ of $n$ points such
that $\delta_0(S,3) \geq 1+\sqrt{3}=2.7321\ldots$ The inequality is tight
for the presented sets.
\end{theorem}
\begin{proof}
Let $P=\{p_0\} \cup P_1 \cup P_2$ be a set of 13 points as
shown in Fig.~\ref{fig:P}\,(left) where $P_1 = \{p_1,p_3,p_5,p_7,p_9,p_{11}\}$
and $P_2 = \{p_2,p_4,p_6,p_8,p_{10},p_{12}\}$.
\begin{figure}[hbtp]
	\centering
	\includegraphics[scale=0.35]{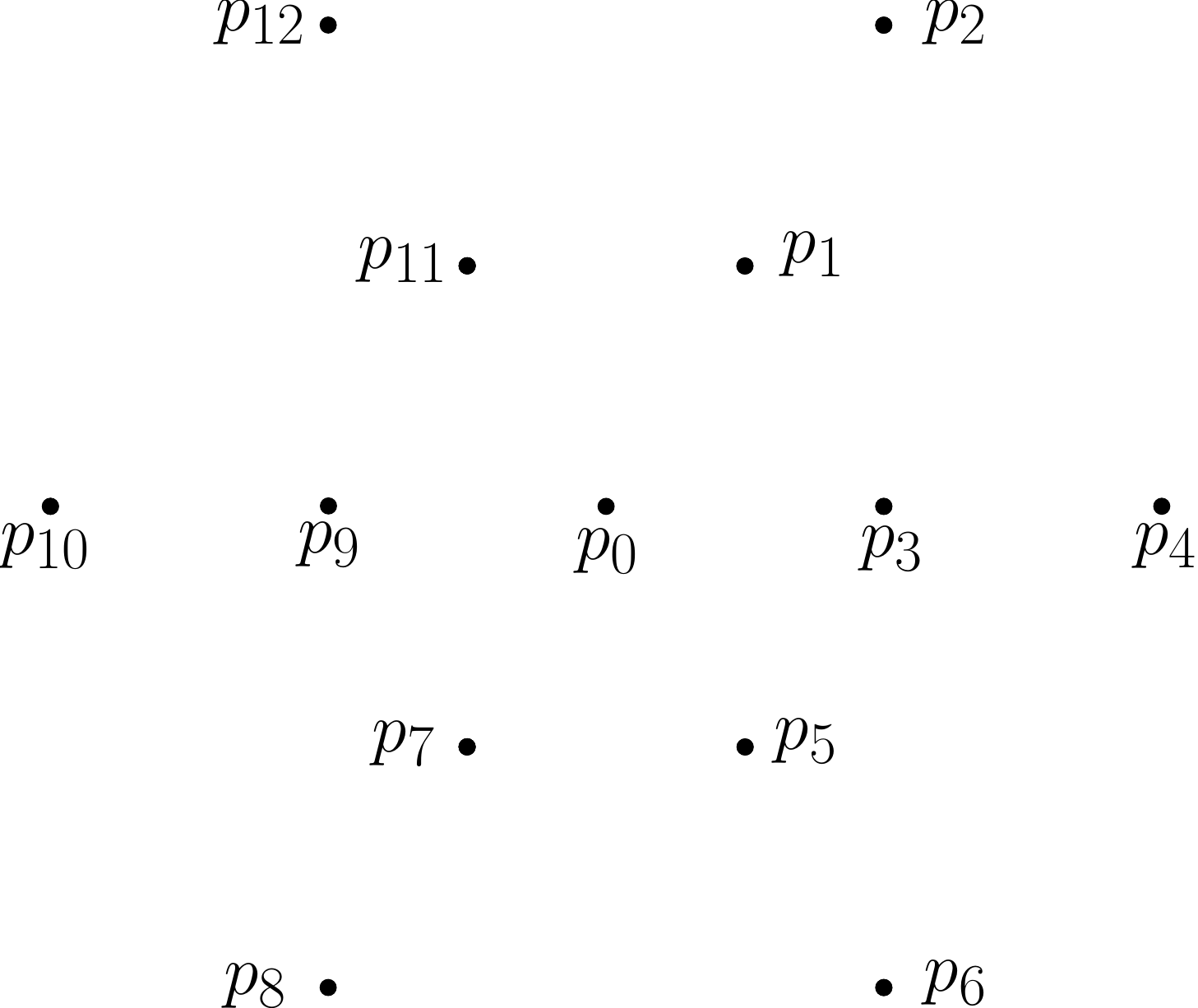}
	\hspace{15mm}
	\includegraphics[scale=0.35]{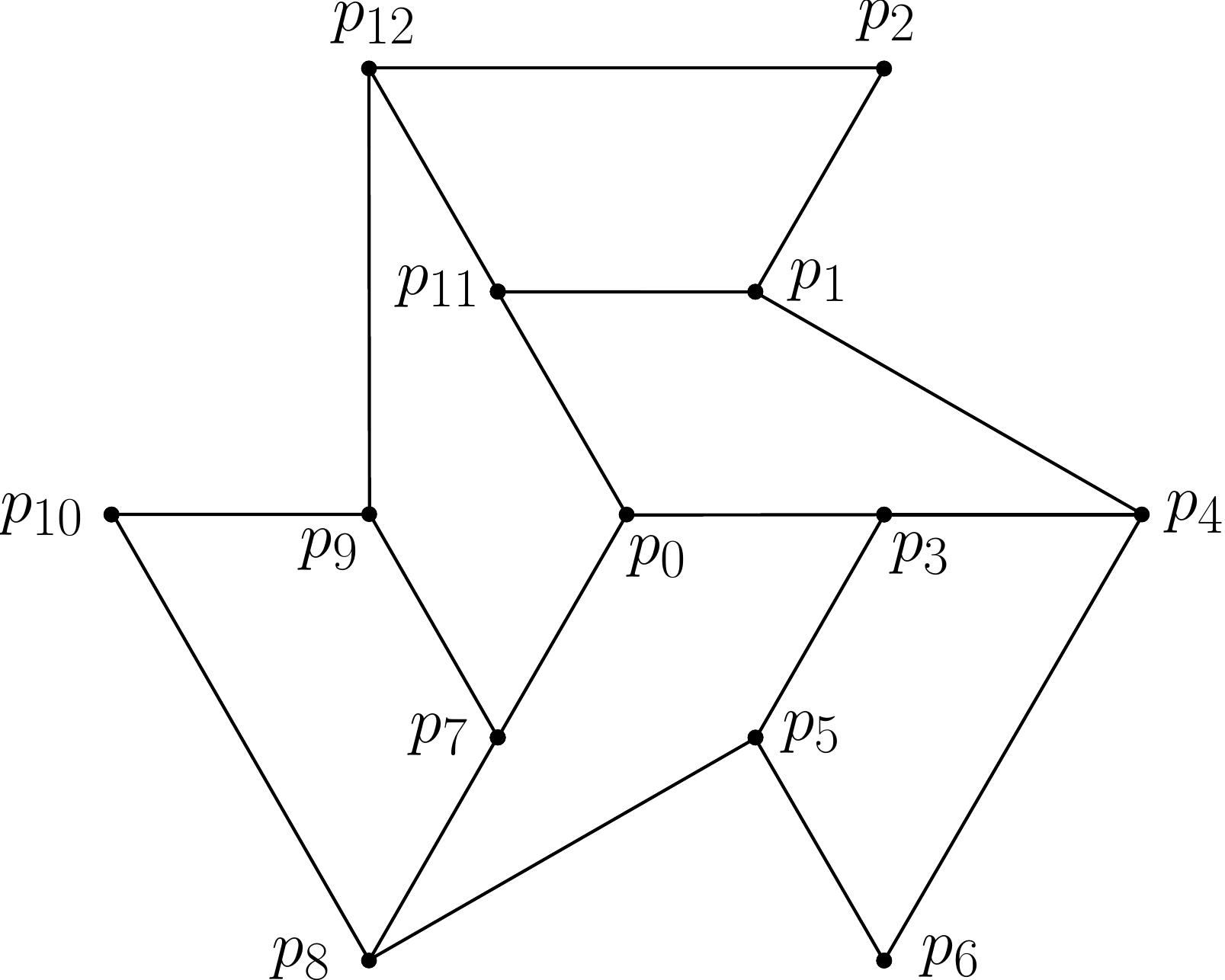}
	\caption{Left: the point set $P=\{p_0,p_1,\ldots,p_{12}\}$;
		some pairwise distances are: $|p_2 p_{12}|=2$,
		$|p_2 p_{3}|=|p_1 p_{5}|=|p_1 p_{12}|=\sqrt{3}$.
		Right: a plane degree~$3$ geometric spanner on $P$ with
		stretch factor  $1+\sqrt{3}$,
		which is achieved by the detours for the point pairs
		$\{p_1,p_3\}$, $\{p_5,p_7\}$ and $\{p_9,p_{11}\}$.} 
	\label{fig:P}
\end{figure}
The points in $P_1$ and $P_2$ lie on the vertices of two regular homothetic hexagons
centered at $p_0$ of radius 1 and 2 respectively. Furthermore, the
points in each of the sets $\{p_2,p_1,p_0,p_7,p_8\}$,
$\{p_4,p_3,p_0,p_9,p_{10}\}$ and $\{p_{12},p_{11},p_0,p_5,p_6 \}$ are collinear. 

We show that $\delta_0(P,3) \geq 1+\sqrt{3}$. Since no edge can contain
a point in its interior,
the point $p_0$ can have connecting edges only with the points from $P_1$.
First, assume that the six edges in
$E = \{p_1p_3, p_3p_5, p_5p_7,p_7p_9, p_9p_{11}, p_{1}p_{11}\}$
are present (see Fig.~\ref{fig:P1}\,(left)).

\begin{figure}[hbtp]
	\centering
	\includegraphics[scale=0.35]{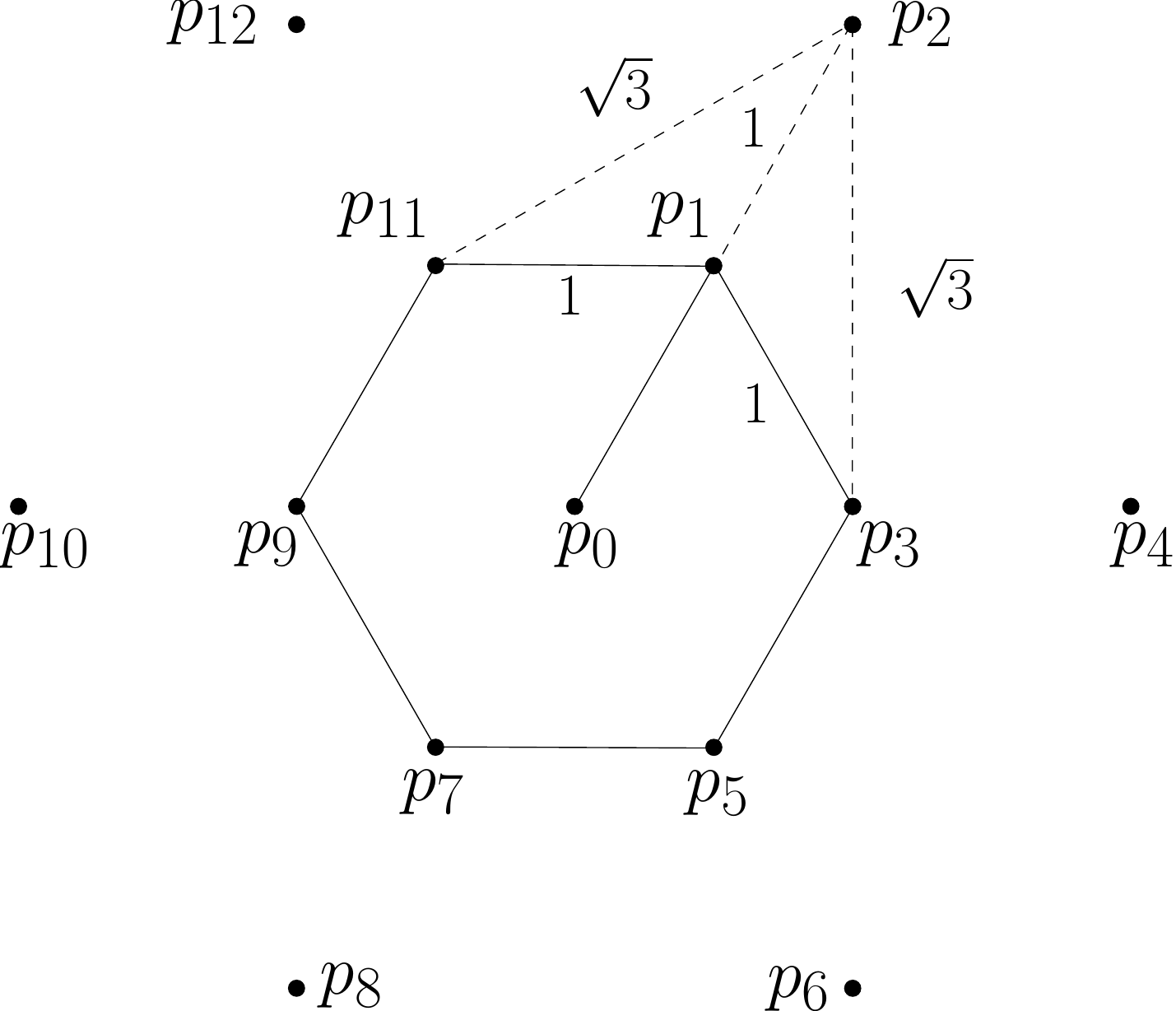}
	\hspace{15mm}
	\includegraphics[scale=0.35]{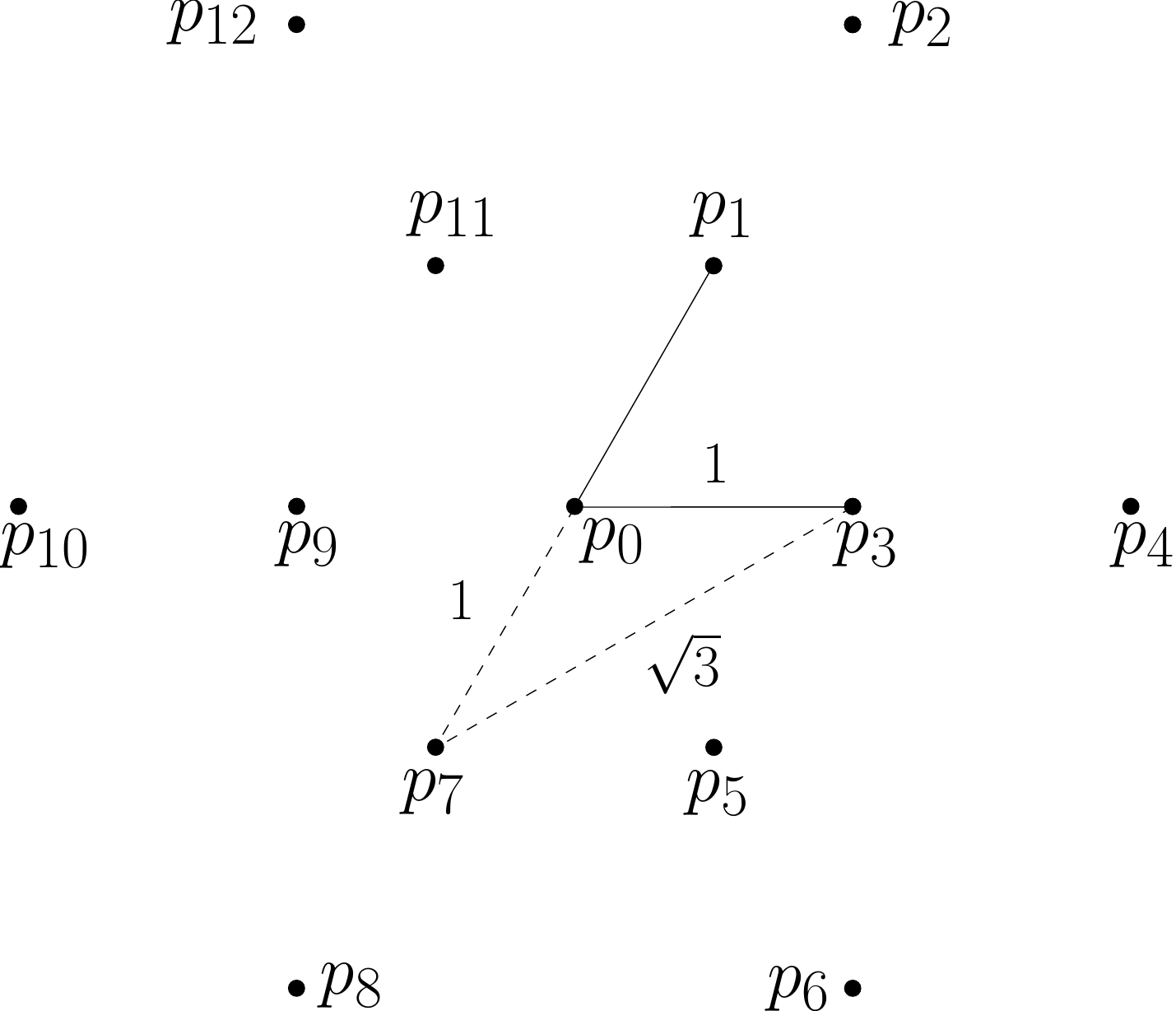}
	\caption{Left: all edges in $E$ are present. Right: \textsc{Case~B}.}
	\label{fig:P1}
\end{figure}

We can also assume that the edge $p_0p_1$ is present since $p_0$ must be connected
to at least one of the points in $P_1$. Observe that now
$\deg(p_1)=3$. In this case,
$$\delta(p_1,p_{2}) \geq \frac{|\rho(1,i,2)|}{ |p_1p_2|} \geq 1 +\sqrt{3},
\text{ where } i \in \{3,11\}. $$  
	Now assume that an edge in $E$, say $p_1p_3$, is missing. Then,
        the following three cases arise depending on $\deg(p_0) \in \{1,2,3\}$. 

\smallskip
\textsc{Case A}:  If $\deg(p_0) = 1$,
          then $$\delta(p_1,p_3) \geq  \frac{|\rho(1,i,3)|}{|p_1p_3|}
          \geq 1+\sqrt{3} \text{ where } i \in \{2,5,4,11\}.$$ 

\smallskip
\textsc{Case B}:  If $\deg(p_0) = 2$, consider the edges
          $p_0p_1, p_0p_3$; see Fig.~\ref{fig:P1}\,(right).
          If $p_0p_1, p_0p_3$ are present
          $\delta(p_0,p_7) \geq |\rho(0,3,7)| / |p_0p_7| =
          1+\sqrt{3}$ else if at least one edge in $\{p_0p_1,
          p_0p_3\}$ is absent then since $p_1p_3$ is absent,
          $\delta(p_1,p_3) \geq  1+\sqrt{3}$ by the same analysis as
          in \textsc{Case} A. 

\smallskip
\textsc{Case C}:  If $\deg(p_0) = 3$, then if at least
           one of the edges $p_0p_1, p_0p_3$ is absent,
           $\delta(p_1,p_3) \geq  1+\sqrt{3}$ as shown in
           \textsc{Case} A. Thus, assume that $p_0p_1,p_0p_3$ are
           present. Now, the following two non-symmetric cases will
           arise. Either $p_0p_5$ is present or $p_0p_7$ is present. 
\begin{figure}[hbtp]
\centering
\includegraphics[scale=0.35]{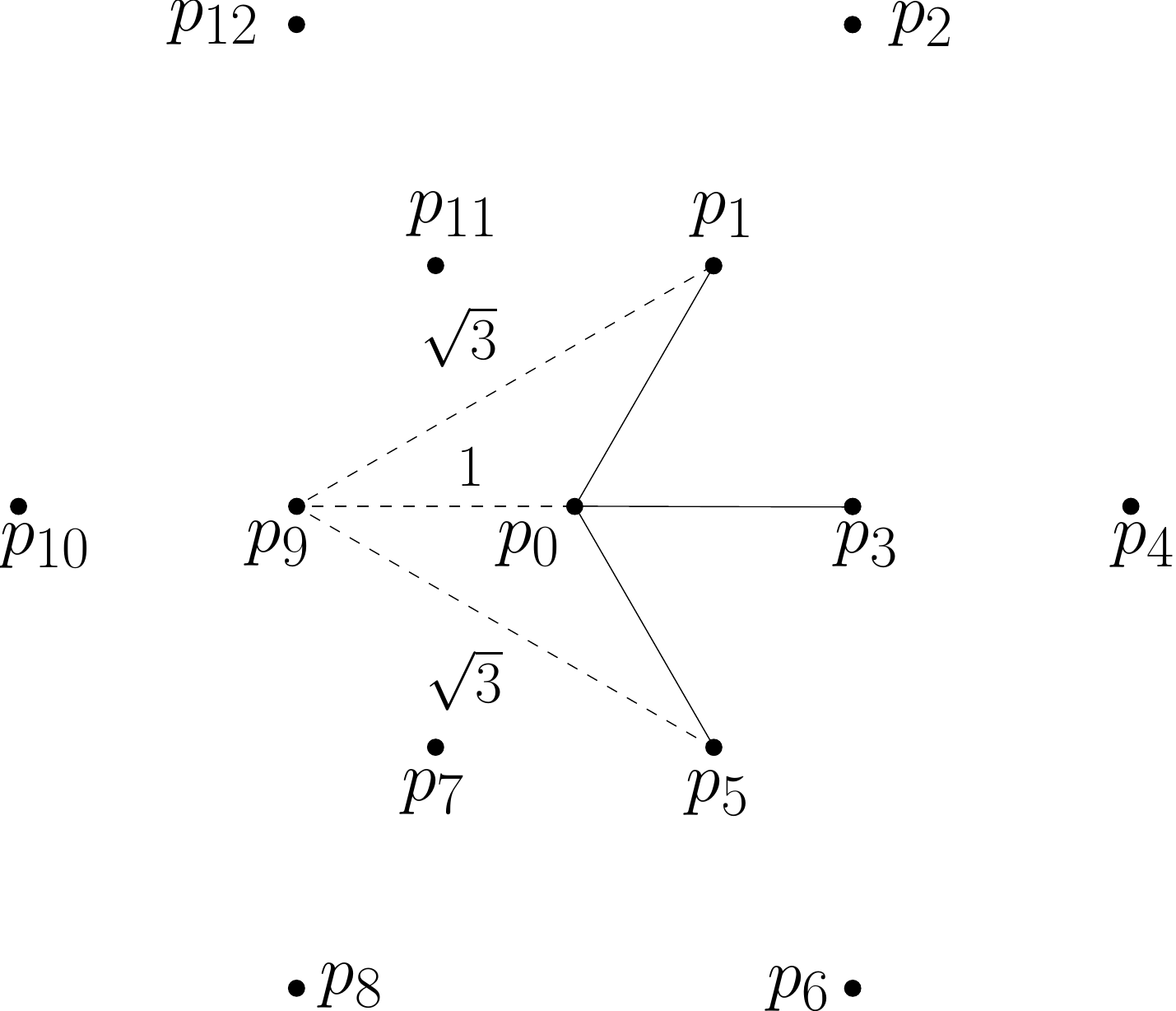}
\hspace{15mm}
\includegraphics[scale=0.35]{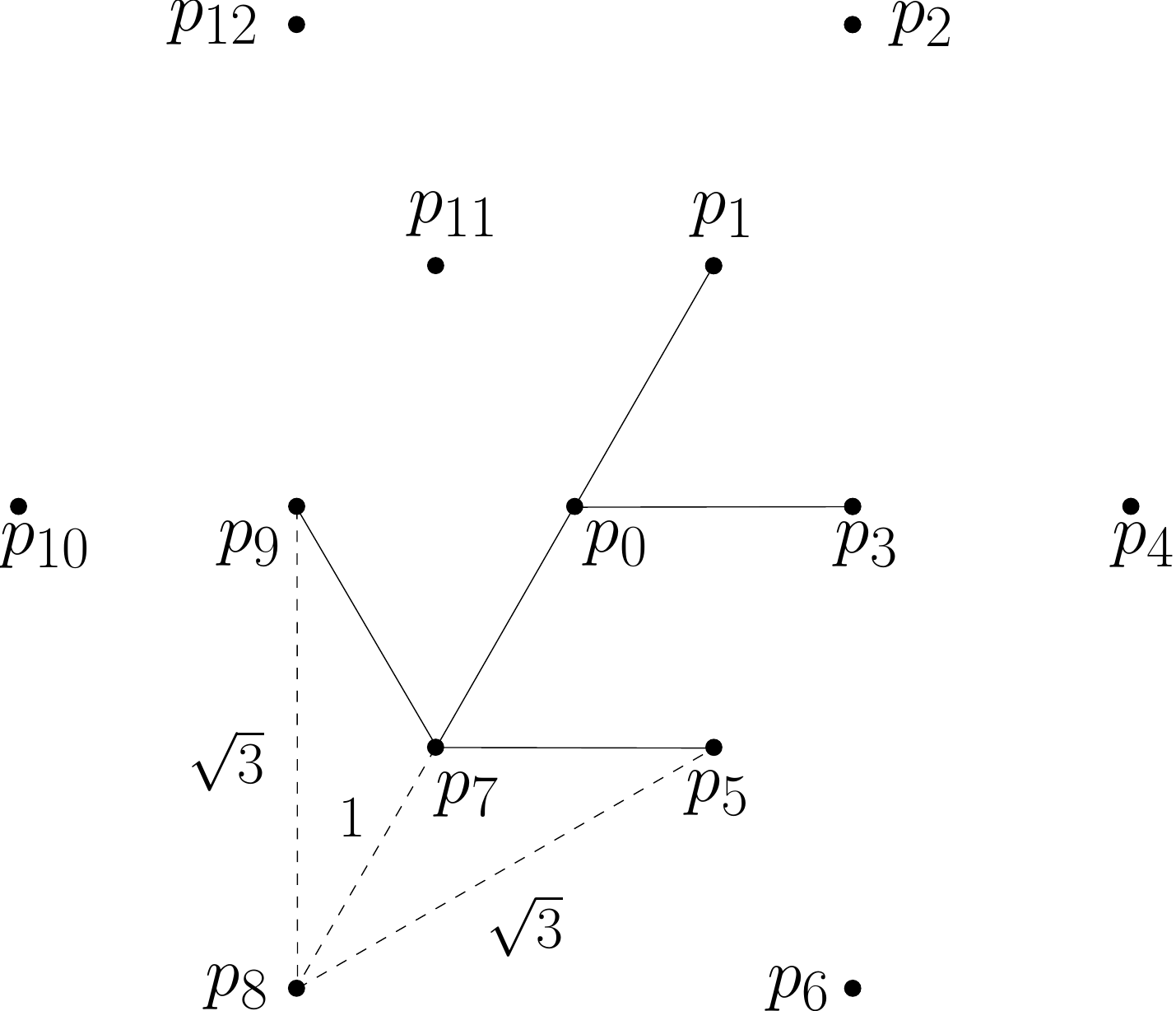}
	\caption{Left: the edges $p_0p_1,p_0p_3,p_0p_5$ are present.
          Right: the edges $p_0p_1,p_0p_3,p_0p_7$ are present.} 
\label{fig:P2}
\end{figure}
	
 If $p_0p_5$ is present (refer to Fig.~\ref{fig:P2}\,(left)) then, 
 $$\delta(p_0,p_9) \geq \frac{|\rho(0,i,9)|}{|p_0p_9|} \geq 1+\sqrt{3},
 \text{ where } i \in \{1,5\}.$$ 
	 
 Now assume that $p_0p_7$ is present (refer to Fig.~\ref{fig:P2}\,(right)).
 Observe that if $p_7p_9$ is absent then, 
 $$\delta(p_7,p_9) \geq \frac{|\rho(7,i,9)|}{|p_7p_9|} \geq 1+\sqrt{3},
 \text{ where } i \in \{8,10,11\}.$$
 Thus, assume that $p_7p_9$ is present. Similarly, assume that $p_5p_7$ is present,
 otherwise
 $$ \delta(p_7,p_5) \geq \frac{|\rho(7,i,5)|}{|p_7p_5|} \geq 1 + \sqrt{3},
 \text{ where } i \in \{3,5,6\}. $$
 Now, as $p_0p_7$, $p_5p_7$ and $p_7p_9$ are present, $\deg(p_7) = 3$. In this case, 
 $$ \delta(p_7,p_8) = \frac{|\rho(7,9,8)|}{|p_7p_8|} \geq 1 + \sqrt{3},
 \text { where } i \in \{5,9\}. $$  

 We have thus just shown that $\delta_0(P,3) \geq 1+\sqrt{3}$.
 For $n\geq 14$, we may assume that $p_0=(0,0)$, $p_3=(1,0)$, and let
 $p_i=(x+i,0)$ for $i=13,\ldots,n-1$, where $x \gg 1$ (\eg, setting $x=100$ suffices);
 finally, let $S = P \cup P'$, where  $P' = \{p_{13},\ldots,p_{n-1}\}$.
If $u ,v \in P \subset S$, then going from $u$ to $v$ via $P'$ is inefficient,
so as shown earlier in this proof, $\delta(u,v) \geq 1+\sqrt{3}$.
Thus, $\delta_0(S,3) \geq 1+\sqrt{3}$, as required.
Moreover, this lower bound is tight for both $P$ and $S$;
see Fig.~\ref{fig:P1}\,(right).
\end{proof}

\paragraph{Remark.} If $\Lambda$ is the infinite hexagonal lattice, it is shown in~\cite{DG16}
that $\delta_0(\Lambda,3) = 1+\sqrt{3}$.

\medskip

We now continue with degree $4$ dilation. We first exhibit a point set
$P$ of $n=6$ points with degree~$4$ dilation $1 + \sqrt{(5-\sqrt{5})/2}$,
and then extend it so to achieve the same lower bound for any larger $n$.
Consider the $6$-element point set $P=\{p_0,\ldots,p_5\}$, where $p_1,\ldots,p_5$
are the vertices of a regular pentagon centered at $p_0$.
\begin{theorem}\label{thm:deg4}
	For every $n \geq6$, there exists a set $S$ of $n$ points such that
        $$\delta_0(S,4) \geq 1 + \sqrt{(5-\sqrt{5})/2} = 2.1755\ldots $$
        The inequality is tight for the presented sets.
\end{theorem}
\begin{proof}
  Assume that $p_1,\ldots,p_5$ lie on a circle of unit radius centered at $p_0$.
  Since $\deg(p_0) \leq 4$, there exists a point $p_i, 1\leq i\leq 5$ such that
  $p_0p_i$ is not present; we may assume that $i=1$; see Fig.~\ref{fig:pentagon}.
 \begin{figure}[hbtp]
 	\centering
 	\includegraphics[scale=0.45]{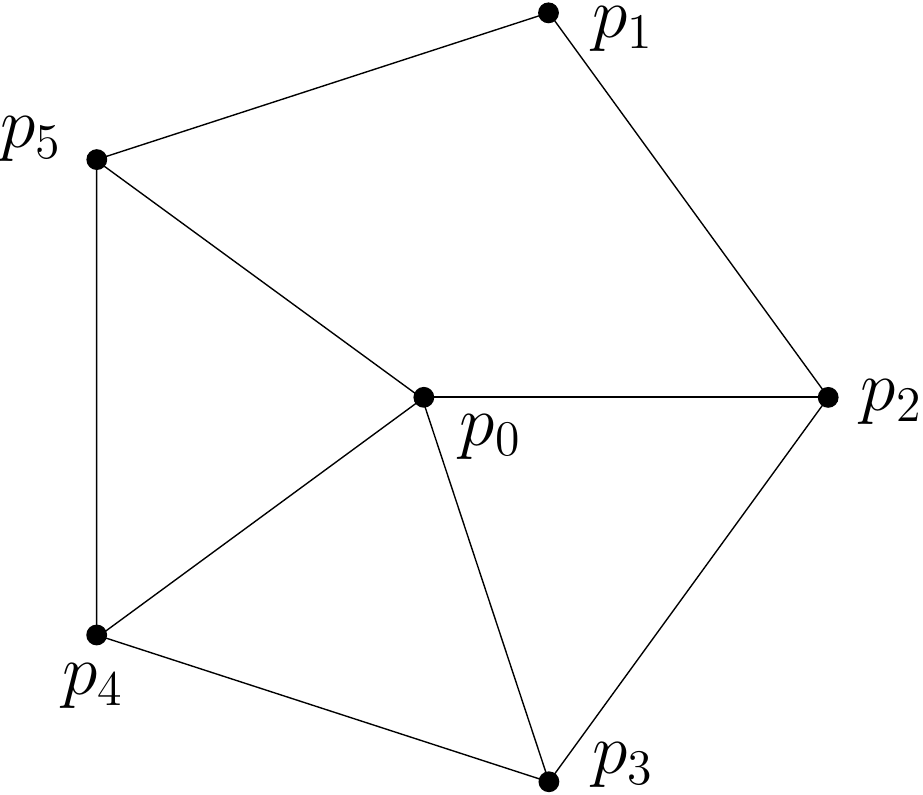}
 	\caption{A plane degree~$4$  geometric graph on the point set
          $\{p_0,\ldots,p_5\}$ that has stretch factor exactly $1 +
          \sqrt{(5-\sqrt{5})/2}$, which is achieved by the detour
          between the pair $p_0,p_1$. } 
 	\label{fig:pentagon}
 \end{figure}
Observe that
  $$ |p_0p_1| = 1 \text{ and }|p_1p_2| = |p_1p_5| =
  \sqrt{1^2+1^2-2\cdot 1 \cdot 1 \cos (2\pi/5)} = \sqrt{(5-\sqrt{5})/2}. $$
Now,
$$\delta(p_0,p_1) \geq \frac{|\rho(0,i,1)|}{|p_0p_1|} \geq 1+
\sqrt{(5-\sqrt{5})/2} = 2.1755\ldots, \text{ where } i\in \{2,5\}. $$ 
	
 Thus, $\delta_0(P,4) \geq 1 + \sqrt{(5-\sqrt{5})/2}$. As in the proof of Theorem~\ref{thm:deg3},
 the aforesaid six points can be used to obtain the same lower bound for any $n \geq 6$.
 
To see that the above lower bound is tight, 
consider the degree~$4$ geometric graph on $P$ in Fig.~\ref{fig:pentagon}
whose stretch factor is exactly that, due to the detour between $p_0,p_1$.
\end{proof}

\section{A lower bound on the dilation of the greedy triangulation}
\label{sec:greedy}

In this section, we present a lower bound on the worst case dilation of the greedy triangulation.
Place four points at the vertices of a unit square $U$, and two other points
in the exterior of $U$ on the vertical line through the center of $U$ and close to the 
lower and upper sides of $U$, as shown in Fig.~\ref{fig:greedy}\,(left).
For any small $\eps>0$, the points can be placed so that the resulting stretch factor
is at least $\delta(p_0,p_3) \geq 2-\eps$.
A modification of this idea gives a slightly better lower bound.

\begin{theorem} \label{thm:greedy}
For every $n\geq 6$, there exists a set $S$ of $n$ points such
that the stretch factor of the greedy triangulation of $S$ is at least $2.0268$. 
\end{theorem}
\begin{proof}
Replace the unit square by a parallelogram $V$ with two horizontal unit sides, unit height
and angle $\alpha \in (\pi/4,\pi/2)$ to be determined, as shown in Fig.~\ref{fig:greedy}\,(right).
Place four points at the vertices of $V$ and two other points  in the exterior of $V$
on the vertical line through the center of the $V$ and close to the
lower and upper side of $V$. First, observe that the greedy triangulation is unique
for this point set. Second, observe that there are two candidate detours
connecting $p_0$ with $p_3$: one of length (slightly longer than)
$1+a$ and one of length (slightly longer than) $2x+b$, where 
$a$ is the length of the slanted side of $V$, $b$ is the length of the short diagonal of $V$,
and $x$ is the horizontal distance between the upper left corner and the center of $V$.
\begin{figure}[hbtp]
	\centering
	\includegraphics[scale=0.6]{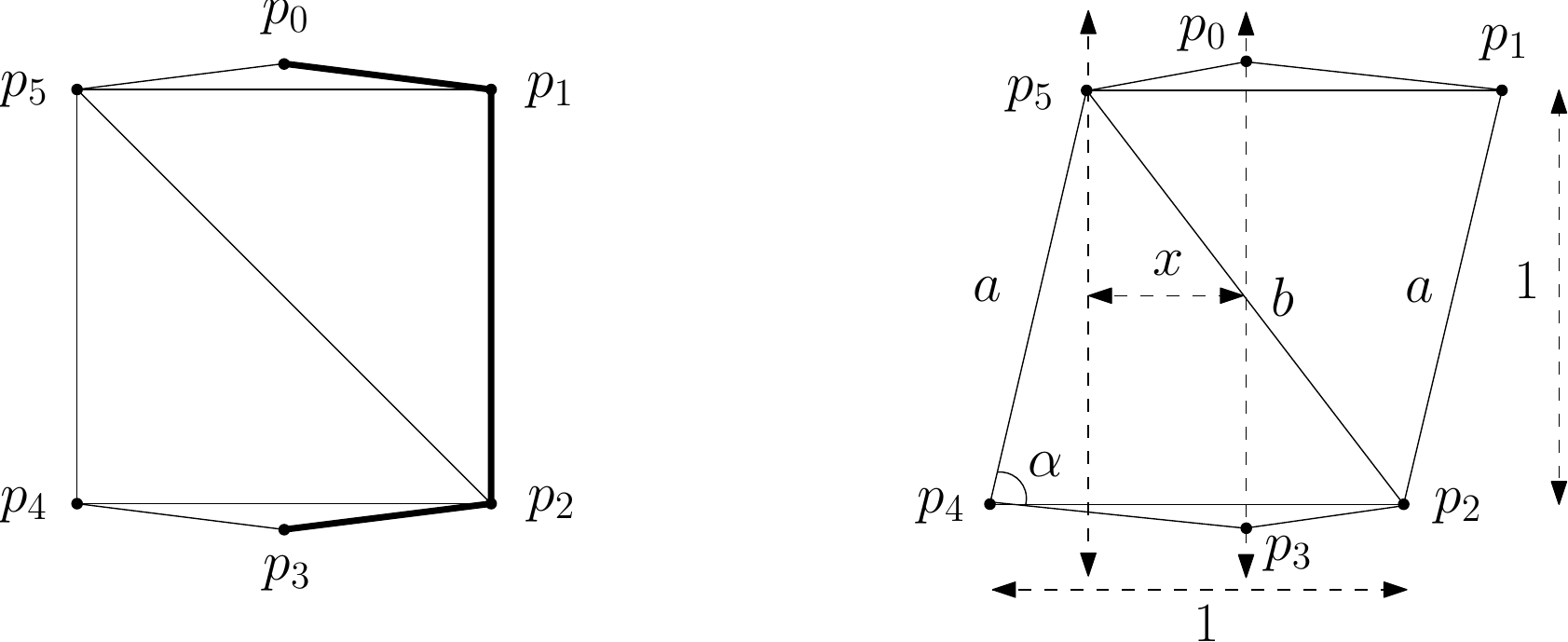}      
      \caption{Greedy triangulation of $6$ points with stretch factors
        $2-\eps$ (left) and $2.0268$ (right).}
	\label{fig:greedy}
\end{figure}

A straightforward calculation gives:
$$ a =\frac{1}{\sin \alpha}, \text{   }
b = \frac{\sqrt{1 + \sin^2 \alpha - 2 \sin \alpha \cos \alpha}}{\sin \alpha},  \text{ and }
x= \frac{1 - \cot \alpha}{2}. $$
$$\text{Let } f(\alpha) =  \min \left( 1 + \frac{1}{\sin \alpha},
1 - \cot \alpha + \frac{\sqrt{1 + \sin^2 \alpha - 2 \sin \alpha \cos \alpha}}{\sin \alpha} \right),
\text{ for } \alpha \in \left(\frac{\pi}{4},\frac{\pi}{2}\right). $$
    
Setting $\alpha= 1.3416$ (\ie, $\alpha= 76.87^\circ$) yields
$$ \delta(p_0,p_3) \geq \max_{\alpha \in (\pi/4,\pi/2)} f(\alpha) \geq f(1.3416) = 2.0268\ldots, $$
as required. As in the proofs of Theorems~\ref{thm:deg3} and~\ref{thm:deg4}, the lower bound
can be extended for every $n \geq 6$ in a straightforward way. 
\end{proof}

\section{Concluding remarks} \label{sec:remarks}

In Section~\ref{sec:lowerbound}, we have shown that any plane spanning graph of the vertices
of a regular $23$-gon requires a stretch factor of 
$(2\sin\frac{2\pi}{23}+ \sin\frac{8\pi}{23})/\sin\frac{11\pi}{23}=1.4308\ldots$
Henceforth, the question of Bose and Smid~\cite[Open Problem 1]{PBMS13} mentioned
in the Introduction can be restated:

\begin{question}
Does there exist a point set $S$ in the Euclidean plane such that
$\delta_0(S) > (2\sin \frac{2\pi}{23} + \sin \frac{8\pi}{23})/\sin \frac{11\pi}{23} = 1.4308\ldots$?
\end{question}

Next in Section~\ref{sec:deg3and4}, it has been shown that there exist point sets that require
degree~$3$ dilation $1+\sqrt{3}=2.7321\ldots$ (Theorem~\ref{thm:deg3})
and degree~$4$ dilation $1 + \sqrt{(5-\sqrt{5})/2}=2.1755\ldots$ (Theorem~\ref{thm:deg4}).
Perhaps these lower bounds can be improved.

\begin{question}
Does there exist a point set in the Euclidean plane that has degree~$3$ dilation
greater than $1+\sqrt{3}$?
Does there exist a point set in the Euclidean plane that has degree~$4$ dilation
greater than $1 + \sqrt{(5-\sqrt{5})/2}$? 
\end{question}

Finally in Section~\ref{sec:greedy}, we show that the stretch factor of the greedy triangulation
is at least $2.0268$, in the worst case.
Perhaps this lower bound is not far from the truth.
Using a computer program we have generated 1000 random uniformly
distributed $n$-element point sets in a unit square for every $n$ in the range
$4 \leq n \leq 250$, and computed the greedy triangulations and
corresponding stretch factors. The highest stretch factor among these was only $1.97$
(as attained for a 168-element point set), and so this suggests the following. 

\begin{question}
  Is the worst case stretch factor of the greedy triangulation attained by
  points in convex position?
\end{question}

Observe that the point set used in the lower bound construction in Theorem~\ref{thm:greedy}
is convex, so it is natural to ask: given a non-convex point set $S$ and a greedy triangulation
of $S$ having stretch factor $\Delta$, does there always exist a convex subset
$S' \subset S$ such that the stretch factor of a greedy triangulation for
$S'$ also equals $\Delta$? The point set $S=\{p_1,\ldots,p_6\}$
illustrated in Fig.~\ref{fig:nonconvex} shows that this is not the case.
It is routine to verify that the stretch factor of the greedy triangulation 
of each convex subset $S' \subset S$ is at most $1.4753\ldots<\Delta = 1.4772\ldots$
\begin{figure}[hbtp]
\centering
\includegraphics[scale=0.39]{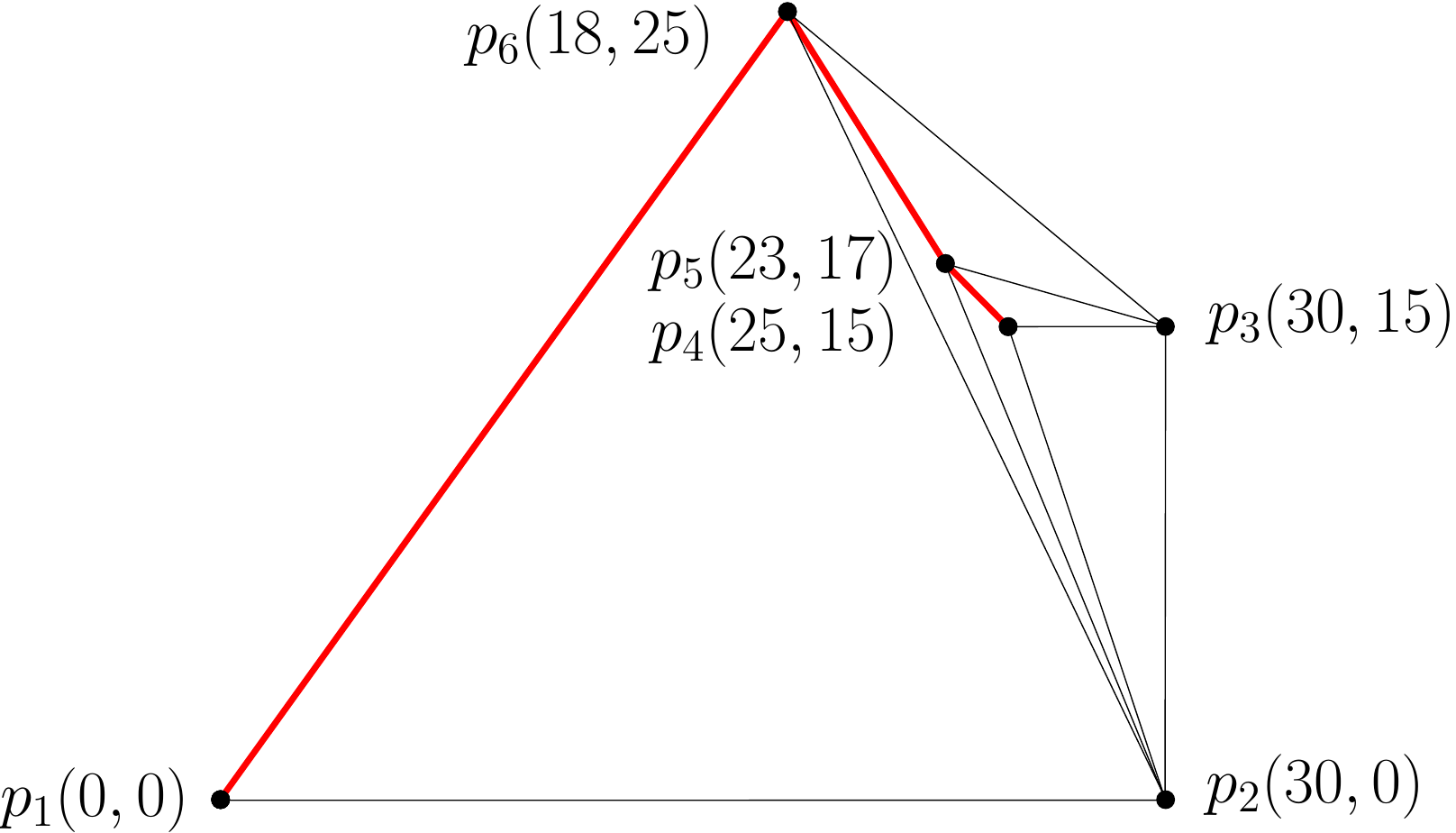}
\hspace{0.2in}
\includegraphics[scale=0.39]{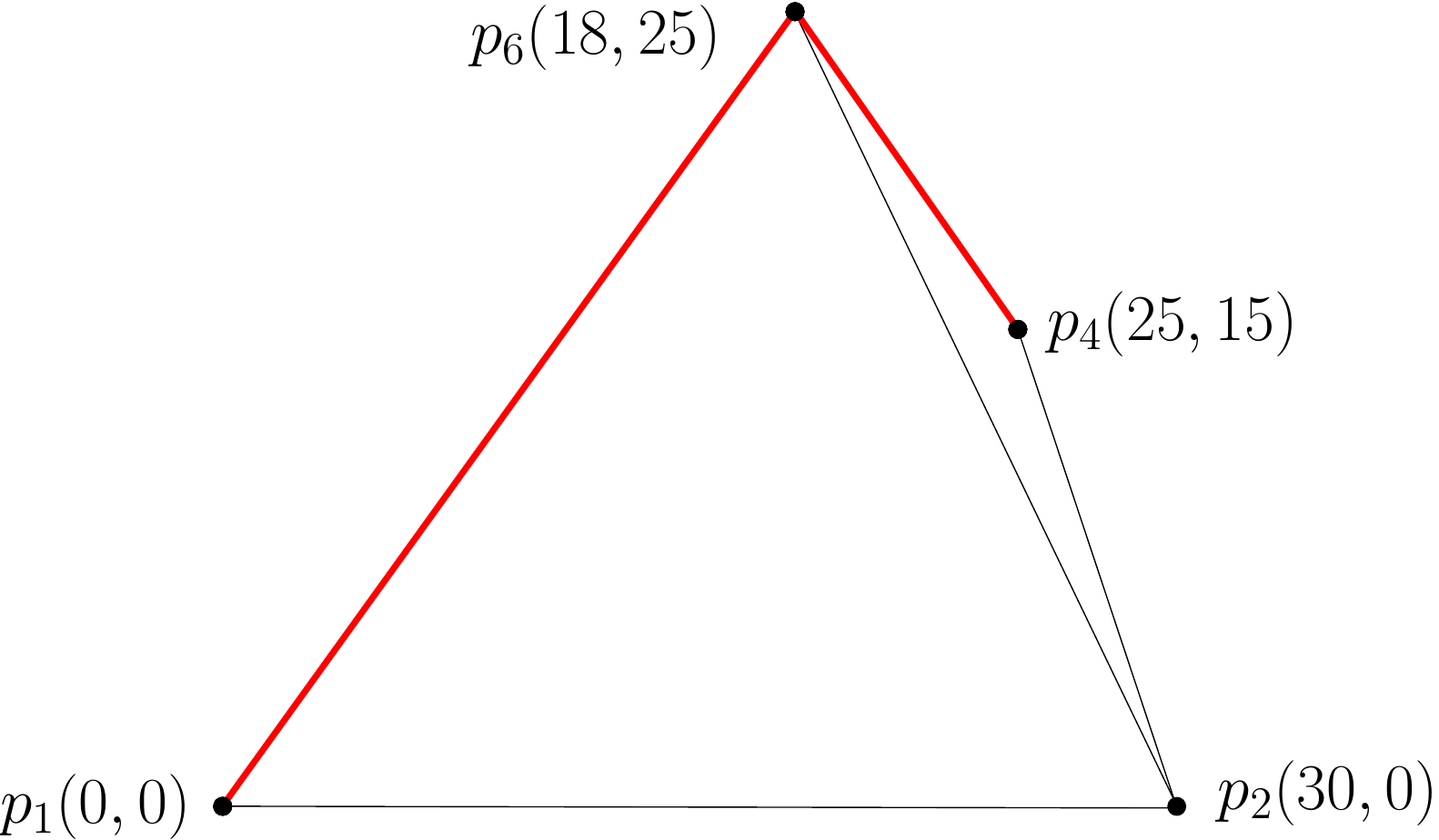}
\caption{Left: greedy triangulation of a set of $6$ points not in convex position with
	stretch factor $\Delta= 1.4772\ldots$ attained by the pair $\{p_1,p_4\}$.
	Right: the largest stretch factor of the greedy triangulation of a convex subset is that
	for the subset $S'=\{p_1,p_2,p_4,p_6\}$; it is attained by the same pair $\{p_1,p_4\}$
	and equals $1.4753\ldots < \Delta$.
	The corresponding shortest paths are drawn in red color.} 
\label{fig:nonconvex}
\end{figure}

\paragraph{Acknowledgment.} We convey our special thanks to an anonymous reviewer for
suggesting an elegant simplification of the case analysis in the proof of Theorem~\ref{thm-s23}.
We express our satisfaction with the software package
\emph{JSXGraph, Dynamic Mathematics with JavaScript} and the OpenMP API,
used in our experiments.

\section*{Appendix}

\paragraph{Source code.} The following parallel C\texttt{++} code is written using
OpenMP in C\texttt{++11} (notice the \texttt{pragma} directives present in the
code). For the set of $N=23$ points, the program ran for approximately
$2$ days on a computer with quad core processor. The
program was compiled with \texttt{g++ 4.9.2}. Please note that older
versions of \texttt{g++} might have issues with OpenMP
support. Following is a correct way of compiling the program. 

\begin{center}
\begin{tcolorbox}[width=10cm,height=10mm]
	\begin{center} 
		\texttt{ { g++ program.cpp -std=c++11 -fopenmp -O3}} \end{center}
\end{tcolorbox}
\end{center}
The number of threads has been set to 4 using the variable
\texttt{numberOfThreads} in  \texttt{main()}. The user may
alter the value of the variable depending on the processor.  

\medskip

Following is the output from the program.

\begin{center}
\begin{tcolorbox}[width=10cm,height=24mm]
\texttt{Execution started...\\
Triangulations checked: 24466267020\\
Dilation: 1.4308143191\\
Time taken: 162829 seconds\\}
\end{tcolorbox}
\end{center}

\begin{lstlisting}
#include <iostream>
#include <cmath>
#include <omp.h>
#include <list>
#include <vector>
#include <chrono>

#define M_PI 3.14159265358979323846

using namespace std;
using namespace chrono;

struct Edge  { unsigned u,v; };
struct gc { unsigned vis_v; Edge oppositeEdge; };

unsigned numberOfThreads;
vector<unsigned long long> countTriangulations;
vector<double> minStretchFactor;
list<unsigned> jobList;

typedef vector<vector<double>> Matrix2D;
vector<Matrix2D> distarrayCollection;

struct Triangulation
{
    unsigned N;
    list<gc> gcList;
    list<Edge> bcList;

    Triangulation(const unsigned numberOfPoints)
    {
        N = numberOfPoints;
        for(unsigned vertexID = 2; vertexID < N-1; vertexID++)
            gcList.push_front({vertexID,{vertexID-1,vertexID+1}});
    }

    void flipgc(const unsigned gchord)
    {
        list<gc>::iterator hold;
        for(auto it = gcList.begin(); it != gcList.end(); it++)
            if(it->vis_v == gchord)
            {
                bcList.push_front(it->oppositeEdge);

                if(next(it,1) != gcList.end())
                    next(it,1)->oppositeEdge.v = it->oppositeEdge.v;

                if(it != gcList.begin())
                    prev(it,1)->oppositeEdge.u = it->oppositeEdge.u;

                hold = it;
                break;
            }
        gcList.erase(hold);
    }
};

inline double distance(const unsigned  p1, const unsigned  p2, const unsigned N)
{
    unsigned absVal = max(p1,p2) - min(p1,p2);
    unsigned lambda = min(absVal,N-absVal);
    return 2*sin((lambda*M_PI)/N);
}


void calculateStretchFactorOfTriangulation(const Triangulation T, const unsigned thread_ID)
{
    double stretchFactor = 0;

    Matrix2D dist = distarrayCollection[thread_ID];

    for(unsigned i = 0; i < T.N; i++)
        for(unsigned j = 0; j < T.N; j++)
            dist[i][j] = (i != j)? INFINITY : 0;

    for(unsigned i = 0; i < T.N-1; i++)
        dist[i][i+1] = dist[i+1][i] = distance(i,i+1,T.N);
    dist[T.N-1][0] = dist[0][T.N-1] = distance(T.N-1,0,T.N);

    for(auto it = T.gcList.begin(); it != T.gcList.end() ; it++)
        dist[0][it->vis_v] = dist[it->vis_v][0] = distance(0,it->vis_v,T.N);

    for(auto it = T.bcList.begin(); it != T.bcList.end() ; it++)
        dist[it->u][it->v] = dist[it->v][it->u] = distance(it->u,it->v,T.N);

    for(unsigned k = 0; k < T.N; k++)
        for(unsigned i = 0; i < T.N; i++)
            for(unsigned j = 0; j < T.N; j++)
                if (dist[i][j] > dist[i][k] + dist[k][j])
                    dist[i][j] = dist[i][k] + dist[k][j];

    for(unsigned i = 0; i < T.N; i++)
        for(unsigned j = 0; j < T.N && i != j; j++)
        {
            double tempRatio = dist[i][j] / distance(i,j,T.N);
            if(tempRatio > stretchFactor)
                stretchFactor = tempRatio;
        }

    if(minStretchFactor[thread_ID] > stretchFactor)
        minStretchFactor[thread_ID] = stretchFactor;
}

void findChildren(Triangulation &T, const unsigned gchord,const unsigned thread_ID)
{
    T.flipgc(gchord);

    countTriangulations[thread_ID]++;
    calculateStretchFactorOfTriangulation(T,thread_ID);

    if(!T.gcList.empty())
        for(auto it = T.gcList.begin(); it != T.gcList.end(); it++)
            if(it->vis_v >= T.bcList.front().u)
            {
                Triangulation *childTriangulation = new Triangulation(T);
                findChildren(*childTriangulation,it->vis_v,thread_ID);
                delete childTriangulation;
            }
}

void thread_job(const unsigned N, const unsigned thread_ID)
{
    Triangulation *childTriangulation;
    unsigned gchord;

    while(true)
    {
        gchord = 0;

        #pragma omp critical
        {
            if(jobList.size() > 0)
            {
                gchord = jobList.front();
                jobList.pop_front();
            }
        }

        if(gchord == 0)
            break;

        childTriangulation = new Triangulation(N);
        findChildren(*childTriangulation,gchord,thread_ID);
        delete childTriangulation;
    }
}

void findAllTriangulations(const unsigned N)
{
    Triangulation rootTriangulation(N);
    countTriangulations[0]++;

    calculateStretchFactorOfTriangulation(rootTriangulation,0);

    #pragma omp parallel for num_threads(numberOfThreads)
    for(unsigned thread_id = 0; thread_id < numberOfThreads; thread_id++)
        thread_job(N,thread_id);

}

void initDataStructures(const unsigned N)
{
    minStretchFactor.assign(numberOfThreads,INFINITY);
    countTriangulations.assign(numberOfThreads,0);

    for(unsigned gc = N-2; gc >= 2; gc--)
        jobList.push_front(gc);

    distarrayCollection.resize(numberOfThreads);
    for(unsigned thread_ID = 0; thread_ID < numberOfThreads; thread_ID++)
    {
        Matrix2D matrix(N);
        for(unsigned pos = 0; pos < N; pos++)
            matrix[pos].resize(N);
        distarrayCollection[thread_ID] = matrix;
    }
}

int main()
{
    unsigned N = 23;
    numberOfThreads = 4;

    double dilation = INFINITY;
    unsigned long long totalNoOfTriangulations = 0;

    cout << "Execution started...\n";

    auto start = system_clock::now();
    initDataStructures(N);
    findAllTriangulations(N);

    for(unsigned thread_ID = 0; thread_ID < numberOfThreads; thread_ID++)
    {
        if( countTriangulations[thread_ID] > 0 && minStretchFactor[thread_ID] < dilation)
            dilation = minStretchFactor[thread_ID];
        totalNoOfTriangulations += countTriangulations[thread_ID];
    }

    printf("Triangulations checked: %lld\n",totalNoOfTriangulations);
    printf("Dilation: %1.10f\n",dilation);

    system_clock::time_point stop = system_clock::now();
    auto duration = duration_cast<seconds>( stop - start ).count();
    cout << "Time taken: " << duration << " seconds" << endl;

    return EXIT_SUCCESS;
}
\end{lstlisting}

\end{document}